\documentclass[11pt,letterpaper]{revtex4-1}
 
\usepackage{fullpage}

\usepackage{amsmath,amssymb,amsthm,amsfonts}
\usepackage[all]{xy}
\usepackage{graphicx,subfigure}
\usepackage{relsize,mdframed,lipsum}

\usepackage{color}
\definecolor{darkblue}{rgb}{0.,0.,0.4}
\definecolor{darkred}{rgb}{0.5,0.,0.}
\usepackage[pdftex,colorlinks=true,linkcolor=darkblue,citecolor=darkred,urlcolor=blue]{hyperref}

\newtheorem{theorem}{Theorem}[section]
\newtheorem{lem}[theorem]{Lemma}
\newtheorem{cor}[theorem]{Corollary}
\theoremstyle{definition} \newtheorem{defn}[theorem]{Definition}
\theoremstyle{definition} \newtheorem{remark}[theorem]{Remark}

\newcommand{\ket}[1]{\left| #1 \right\rangle}
\newcommand{\bra}[1]{\left\langle #1 \right|}
\newcommand{\rann}{r_\text{ann}}

\begin{document}
\title{An invariant of topologically ordered states under local unitary transformations}
\author{Jeongwan Haah}
\email{jwhaah@mit.edu}
%\address % USE THIS LINE FOR AMSART
\affiliation % OR THIS FOR REVTEX4
{Department of Physics, Massachusetts Institute of Technology, Cambridge, Massachusetts}
\date{11 October 2015}
\begin{abstract}
For an anyon model in two spatial dimensions described by a modular tensor category,
the topological S-matrix encodes the mutual braiding statistics,
the quantum dimensions, and the fusion rules of anyons.
It is nontrivial whether one can compute the S-matrix from a single ground state wave function.
Here, we define a class of Hamiltonians consisting of local commuting projectors 
and an associated matrix that is invariant under local unitary transformations.
We argue that the invariant is equivalent to the topological S-matrix.
The definition does not require degeneracy of the ground state.
We prove that the invariant depends on the state only,
in the sense that it can be computed by any Hamiltonian in the class
of which the state is a ground state.
As a corollary, we prove
that any local quantum circuit that connects two ground states of quantum double models (discrete gauge theories)
with non-isomorphic abelian groups, must have depth that is at least linear in the system's diameter.

As a tool for the proof,
a manifestly Hamiltonian-independent notion of locally invisible operators is introduced.
This gives a sufficient condition for a many-body state
not to be generated from a product state by any small depth quantum circuit;
this is a many-body entanglement witness.
\end{abstract}
\maketitle

\section{Introduction}

A classification problem is defined by two ingredients,
a set of objects and an equivalence relation on it.
The problem is solved if one knows a complete list of equivalence classes and 
can tell which equivalence class each object represents.
A natural approach is of course to find \emph{invariants}
assigned for each equivalence class and a method to compute them
given a representative of the equivalence class.
The classification is completed when a complete set of invariants is found
such that distinct equivalence classes have distinct sets of invariants.

For gapped Hamiltonians,
the equivalence relation is given by continuous paths in a space of gapped local Hamiltonians.
An equivalence class is a quantum phase of matter.
The quasi-adiabatic evolution~\cite{HastingsWen2005Quasi-adiabatic,Osborne2007}
gives us a new insight to this problem.
A gapped path of Hamiltonians defines a unitary transformation $U$
that is generated by a sum of (quasi-)local operators
and relates the ground state subspace on one end of the path
to that on the other end.
Because of the locality of the generator of $U$,
it is well approximated by a product of layers of non-overlapping local unitary transformations,
called a local quantum circuit,
where the number of layers or the depth of the quantum circuit
is much smaller than the system size.
Conversely, a small-depth local quantum circuit defines a piece-wise smooth
path in the gapped local Hamiltonian space,
as any local unitary operator can be written as the exponential of a local operator.
Thus, a gapped Hamiltonian path and a quantum circuit are intuitively equivalent,
though proving it may involve very nontrivial approximation analyses depending on purposes.
We will discuss only the quantum circuits in this paper for their simplicity.

The intuitive correspondence between a gapped Hamiltonian path and a small-depth quantum circuit
suggests that all information about the quantum phase of matter
is contained in the ground state wave function,
since the equivalence relation by quantum circuits can be tested
for given wave functions.
Indeed, for two-dimensional topologically ordered systems
it has been proposed that 
the total quantum dimension~\cite{KitaevPreskill2006Topological,LevinWen2006Detecting}
and the topological S-matrix~\cite{ZhangGroverTurnerEtAl2011Quasiparticle}
can be computed from the ground state wave functions.
The total quantum dimension $\mathcal D$ is shown to be equal to the exponential of the universal constant correction,
called the topological entanglement entropy,
to the area law of entanglement entropy.
Namely, if $A$ is a contractible disk of circumference $L$, 
which is much larger than the correlation length,
then the entanglement entropy of the ground state for $A$ obeys
\[
 S_A = \alpha L - \log \mathcal D + \cdots ,
\]
where $\alpha$ is a model-dependent constant.
Thus it can be computed from a single wave function.
The reason that the entanglement entropy gets negative correction may be
intuitively understood in the string condensate picture~\cite{LevinWen2005String-net}
as the requirement that strings be closed should reduce the entropy.
It has been argued that the total quantum dimension is the only information
one can extract from entanglement entropies of a single wave function~\cite{FlammiaHammaHughesWen2009TEE}.
Even so, it was shown that from the full set of ground states on a torus 
one can extract the topological S-matrix~\cite{ZhangGroverTurnerEtAl2011Quasiparticle},
which is a finer quantity of the phase of matter than the total quantum dimension.
Based on these ideas, numerical computation was 
performed~\cite{DepenbrockMcCullochSchollwoeck2012Kagome, JiangWangBalents2012},
measuring the topological entanglement entropy,
to give an evidence that a highly frustrated magnetic system exhibits topological order.

The argument for the robustness or the invariance of these quantities, however, is indirect.
The argument of Kitaev and Preskill~\cite{KitaevPreskill2006Topological}
that there exists a universal constant correction
is invalid unless one assumes some homogeneity
because of the Bravyi's counterexample~\footnote{S. Bravyi, private communication.}.
Bravyi pointed out that the so-called one-dimensional cluster state, if cleverly embedded
into the plane, gives a nonzero contribution to the topological entanglement entropy,
although the cluster state can be created from a product state by a quantum circuit of depth two.
Kim~\cite{Kim2012Perturbative}
obtains a bound on the first order perturbation of the topological entanglement entropy,
assuming certain homogeneity and an extra criterion on conditional mutual information for various regions.
%The homogeneity has not been elaborated in any form as far as we know.
The topological S-matrix of Zhang {\it et al.}~\cite{ZhangGroverTurnerEtAl2011Quasiparticle}
seems robust, but it relies on an assumption that the system is described by a modular tensor category.
On the mathematics side,
there is Ocneanu's rigidity theorem proving that any modular tensor category cannot be deformed
except for trivial basis changes~\cite{EtingofNikshychOstrik2005}
(See also Appendix~E.6 of \cite{Kitaev2006Anyons}),
but it is nontrivial when and how a two-dimensional Hamiltonian
corresponds to a modular tensor category.

Our understanding of invariants is reflected in the status of the following problem.
If two wave functions are ground states of distinct phases of matter,
how deep a local quantum circuit must be
in order to transform one wave function to another?
This question has been addressed and partially solved by 
Bravyi, Hastings, and Verstraete~\cite{BravyiHastingsVerstraete2006generation};
see also \cite{EisertOsborne2006}.
They showed that if there are two or more degenerate
ground states that are locally indistinguishable,
then this local indistinguishability is an invariant.
However, the local indistinguishability can only tell trivial product states from topological states.
Also, this is applicable, tautologically, when there are two or more states.
A topological model can be defined on a sphere,
in which case the local indistinguishability is of no use
since there is a unique ground state.
K\"onig and Pastawski~\cite{KoenigPastawski2014Generating}
have realized this limitation.

\begin{figure*}[tb]
\subfigure[]{\includegraphics[width=.3\textwidth]{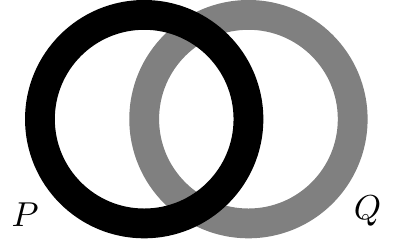}}
\qquad
\subfigure[]{\includegraphics[width=.3\textwidth]{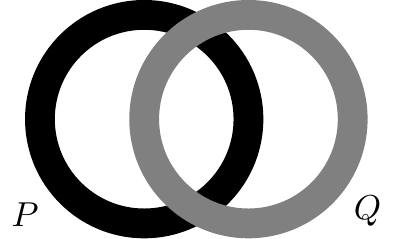}}
\qquad
\subfigure[]{\includegraphics[width=.3\textwidth]{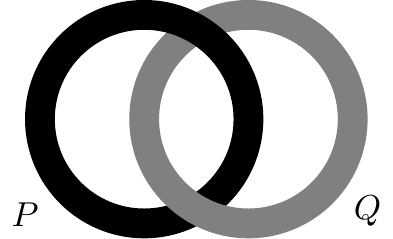}}
\caption{Two annuli of the same radius on a large plane 
are intersecting at two diamond-like regions.
The distance between the two diamonds are comparable with the radius of the annuli.
$P$ and $Q$ denote operators supported on the left and right annulus, respectively.
In (a), the usual product $P \cdot Q$ is depicted. The operator on the left is drawn closer to the reader.
In (b), the product $Q \cdot P$ is shown.
In (c), the twist product $P \infty Q$ is depicted. The order of the multiplication 
is reversed for operator components in the bottom region. See Eq.~\eqref{eq:def-twist}.
$\tilde S_{ab}$ is defined by $\bra \psi \pi_a \infty \pi_b \ket \psi$ where $\pi_a$ and $\pi_b$ are fundamental projector of the logical algebra
on the left and right annulus, respectively.}
\label{fig:annuli}
\end{figure*}

Here, we show that for topologically ordered systems
one can define a matrix $\tilde S$, and that it is an invariant of the state under local unitary transformations.
The systems we consider are restricted to those with unfrustrated commuting projector Hamiltonians
satisfying certain natural conditions which are elaborated later.
We argue that this matrix must match the topological S-matrix if the system is described by a modular tensor category.
The invariance proof directly implies that any transformation quantum circuit between states with distinct $\tilde S$ matrices
has depth that is at least linear in the system's diameter.
Although our argument is certainly motivated by the modular tensor category description,
our treatment is independent from it.
Rather, our result should be read as a justification of the modular category description
for certain lattice systems.
In addition, our treatment solely requires a large disk-like region,
and does not depend on the boundary conditions.
This might shed some light 
on how to compute a finer invariant, the S matrix, given a single wave function.
This is potentially important in numerics for identifying a quantum phase of matter in 2D
because the density matrix renormalization group method is systematically biased
to select a particular ground state even if there is degeneracy~\cite{JiangWangBalents2012},
which might prohibit obtaining the full ground state subspace.

We assume that the system is governed by an unfrustrated 
local commuting projector Hamiltonian with two extra conditions.
The first one is the local topological order condition,
which is natural for any anyon models~\cite{BravyiHastingsMichalakis2010stability}.
This condition asserts that the local reduced density matrix on a region $M$ of
a ground state wave function is inferred from the terms of the Hamiltonian around $M$.
The second one is a new condition, which we call \emph{stable logical algebra condition},
asserting that the algebra of all string operators
on an annulus is independent of the thickness of the annulus.
A motivation for this condition is that
the interaction among anyons is determined by
the topology of the paths they have moved along.
It seems that the second condition closely related to the finiteness
of the number of topological particle types in two dimensions;
we show that 
this condition fails to be true when there are infinite number of particle types.

The $\tilde S$ is defined by particle type projection operators $\pi_a$
supported on two annuli of comparable size as shown in Figure~\ref{fig:annuli}.
Here, $\pi_a$ projects onto a state where
a definite particle is present in the inner disk of the annulus.
$(a,b)$-entry of $\tilde S$ is the expectation value of 
the {\em linking} of $\pi_a$ on the left annulus and $\pi_b$ on the right annulus.
This is almost identical to the usual definition of the topological S-matrix.
The difference is that we use the particle type projection operators
instead of string operators that carry particles of definite types.
Our approach is advantageous because we can avoid any difficulty
in defining the string operators associated with definite particle types.

Our $\tilde S = \tilde S(\ket \psi)$ matrix is independent of Hamiltonians
in the sense that whenever $\ket \psi$ is a common ground state of $H_1$
and $H_2$, where each of $H_1$ and $H_2$ satisfies our two conditions,
the result $\tilde S$ is the same.
Central to the proof of this is a notion of {\em locally invisible} operators.
They are invisible since there is no way to tell 
whether the operator has acted on, by looking at a small region.
This is an abstraction of the string operators that transport anyons.
It is important that the locally invisible operators are defined
without any reference to a Hamiltonian.

It turns out that the existence of the locally invisible operators
is a signature of topological order or ``long-range'' entanglement.
That is, the locally invisible operators are {\em many-body entanglement witnesses.}
The locally invisible operators need not wrap around a topologically nontrivial loop;
they can distinguish a topologically ordered state from a product state
even though they are supported on a large but finite disk.
Section~\ref{sec:locally-invisible} where we study locally invisible operators
can be read independent of the rest of this paper.

The main text is organized as follows.
In Section~\ref{sec:S-matrix} we define the particle type projector operators $\pi_a$
and the matrix $\tilde S$ for a Hamiltonian $H$,
and argue that having $\tilde S$ is equivalent to having the topological S-matrix.
We will compute $\tilde S$ for the simplest model, the toric code~\cite{Kitaev2003Fault-tolerant}.
Section~\ref{sec:locally-invisible} is devoted to the definition and properties of locally invisible operators.
An unanswered question of Ref.~\cite{KoenigPastawski2014Generating} will be answered,
on how deep a quantum circuit must be in order to create a toric code state on a sphere.
The last section~\ref{sec:invariant} contains our main theorem.
We remark that for abelian discrete gauge theories $\tilde S$
is a complete invariant.
The section is divided into several subsections to clarify and explain our technical conditions.

\subsection*{Convention and terminology}

We consider lattices where finite dimensional degrees of freedom (qudits)
are placed at each site.
We say an operator $O$ is {\bf supported} on a region $D$
if $O$ acts by identity on the complement of $D$.
The {\bf support} of $O$ is the minimal region $D$ on which $O$ is supported.
The {\bf subspace support} of a hermitian operator $\rho$ is
the linear space spanned by the eigenvectors of $\rho$ with nonzero eigenvalues.
(The projection onto the subspace support of the density operator $\rho$
is also known as {\it support projection} in operator algebra.)
An {\bf annulus} is the bounded region between two concentric circles on a plane.
We say the annulus is of radius $\rann$ and thickness $t>0$ if
the radius of the inner circle is $\rann - t$, and the outer $\rann +t$.
A {\bf (local) quantum circuit} is an ordered product of layers of local unitary operators of non-overlapping support.
Every local unitary operator in a layer acts on at most two qudits 
that are nearest neighbors of each other.
The {\bf depth} of a quantum circuit is the number of layers in it.
The {\bf range} of the quantum circuit $W$ is the minimum distance $R$
such that for every operator $O$ supported on a disk of radius $r$
the conjugated operator $WOW^\dagger$ is supported on a disk of radius $r+R$ for any $r$.
Since local unitary operator in a quantum circuit acts on a pair of qudits separated by distance $1$,
the range of the quantum circuit is always less than or equal to its depth.
The {\bf interaction range} of a local Hamiltonian 
is the distance $w$ such that any term in the Hamiltonian is supported on a disk of diameter $w$.

\section{Topological S-matrix}
\label{sec:S-matrix}

The topological S-matrix $S$ contains information on the mutual braiding interaction between far separated excitations.
The rows and columns of $S$ are indexed by types of excitations
and $(a,b)$ entry is the quantum mechanical amplitude of the process,
where a particle-anti-particle pair of type $a$ and another pair of type $b$
are created and separated, the particle $a$ is moved around the particle $b$ (braiding),
and then the particle-anti-particle pairs are annihilated.
If the braiding is sufficiently nontrivial, then $S$ 
with a proper normalization is a unitary matrix.~\cite{BakalovKirillov2001}

This explanation contains essential physical intuition of the S-matrix,
but this, as it is stated, cannot be a definition
since we do not specify how particle types are defined.
The particle types must be a quantum number (conserved quantity)
that is invariant under certain local operators.
The class of allowed local operators that preserve the quantum number
is usually determined by the symmetry of the system.
For example, the total angular momentum is invariant under
the action of rotation symmetry on the physical state of a particle,
and hence is one of defining quantum numbers of the particle type, the spin.
However, in topological systems, which is of interest in our discussion,
there is no symmetry restriction.
A conservation law in this unrestricted situation is termed as a superselection rule~\cite{WickWightmanWigner1952},
and the particle types are identified with the superselection sector that the particle represents~\cite{Preskill2004TQCLecture}.

\subsection{Particle types}
\label{sec:particle-types}

To be concrete, let us restrict our discussion to two-dimensional lattices 
with finite dimensional degrees of freedom placed at each site,
and Hilbert space is given by the tensor product of the local degrees of freedom.
Consider a local Hamiltonian
\[
 H = - \sum_j h_j
\]
and localized excitation $e$ that is separated from the rest excitations by a large distance.
We assume that every term of $H$ is a projector commuting with any other term $[h_j,h_{j'}] = 0$,
and on the ground state $\ket \psi$ every term of $H$ is minimized $h_j \ket \psi = \ket \psi$,
i.e., $H$ is not frustrated.
The state $\ket{\psi'}$ in which the excitation $e$ is separated from the rest
can be described by the condition
\begin{equation}
 h_a \ket{ \psi'} = \ket{ \psi' }
 \label{eq:vacuum-condition}
\end{equation}
whenever $h_a$ is supported on an annulus $A$ circling around $e$.

What operators should we consider in order to learn about the superselection sector that $e$ represents?
Certainly, we should consider any operators acting on the inner disk $D$ of the annulus.
They form a matrix algebra
\[
 \mathrm{Mat(D)}.
\]
These operators will create and annihilate excitations on the inner disk, but will not change
the fact that excitations outside the annulus are far separated from those on the inner disk.
In addition, we should consider Aharanov-Bohm interference measurements about $e$,
by letting various particles encircle $e$ along the annulus.
The operators of this class are supported on the annulus and commute with terms of the Hamiltonian
since the overall effect of the measurement should not introduce any new excitation.
So they belong to an algebra
\[
 \mathcal A = \{ O \text{ on } A ~|~ [O, h_j] = 0 \text{ for all $h_j$} \}.
\]
Since $h_j$ are hermitian, $\mathcal A$ is in fact a $C^*$-algebra.

The elements of $\mathcal A$ do not act faithfully on the state $\ket{\psi'}$,
because it is defined regardlessly of Eq.~\eqref{eq:vacuum-condition}.
For instance, $N = (1-h_a)$ with $h_a$ supported on the annulus
is an element of $\mathcal A$, but acts on $\ket{\psi'}$ as zero;
any two operators of $\mathcal A$ that differ by a \emph{null} operator $N$
should have the same result on $\ket{\psi'}$.
In view of the interference thought-experiment, the null operator corresponds
to flux measurement along a small loop that is entirely contained in the annulus,
followed by the postselection onto a nonzero flux value.
Since the vacuum condition Eq.~\eqref{eq:vacuum-condition} ensures
that the flux is zero along such small loops, this outcome has necessarily zero probability.
Thus, it is appropriate to factor out $\mathcal A$ by
\begin{equation}
 \mathcal N = \left \{ O \in \mathcal A ~\middle|~O \left(\prod_{\mathrm{supp}(h_a) \subseteq A_{+w} } h_a \right)= 0 \right\}  ,
 \label{eq:def-nullOP}
\end{equation}
the set of all null operators on $\ket{\psi'}$,
where $\mathrm{supp}$ means the support and $A_{+w}$ is the annulus enlarged from $A$ by the interaction range $w$.

Now, the particle types are conserved quantities under the operations of
\[
 \mathcal C =  \mathcal A / \mathcal N \otimes \mathrm{Mat}(D) .
\]
This set of operators forms a $C^*$-algebra, and for a thick but finite annulus,
it is finite dimensional over complex numbers.
Recall that any finite dimensional $C^*$-algebra $\mathcal C$ is isomorphic to
a direct sum of full matrix algebras (simple algebras), and there exist projectors $\pi_a$ in the center of $\mathcal C$
onto the simple algebra components.
The projectors, which we call {\bf fundamental projectors},
are uniquely determined by the algebra $\mathcal C$ from the conditions
\begin{align}
 \mathcal C &= \bigoplus_a \pi_a \mathcal C \text{ where $\pi_a \mathcal C$ is simple,} \label{eq:fundamentalprojectors}\\
 \pi_a &= (\pi_a)^\dagger = (\pi_a)^2 \text{ is in the center of $\mathcal C$.} \nonumber 
\end{align}
We conclude that the fundamental projectors in the center of $\mathcal C$ is precisely
the conserved operators.
Therefore, we \emph{identify} the particle types of excitations, or anyon labels,
that can be supported in the disk $D$ with the fundamental projectors of $\mathcal C$.
An immediate observation is that any central element of $\mathcal C$
should have the identity component on $\mathrm{Mat}(D)$,
since any full matrix algebra has trivial center consisting of scalars.
This means that the fundamental projectors of $\mathcal C$ is actually the fundamental projectors of 
the smaller algebra $\mathcal A / \mathcal N$, which we call the {\bf logical algebra}, on the annulus.

We do \emph{not} claim that we have derived the correspondence rigorously between the particle types
and the fundamental projectors of the algebra $\mathcal C$.
Rather, mathematically, our argument should be regarded as a motivation to \emph{define} particle types.
% As we will see later, our definition has an important consequence that
% it depends on the ground state only, and does not depend on the Hamiltonian, to some extent.

Remark that if there are completely disentangled qudits in the ground state,
the logical algebra is invariant upon the removal of those qudits.
Specifically, suppose $h_{a_0} = \ket 0 \bra 0$ was a single-qudit operator of rank 1 acting on the qudit $a_0$.
Any operator $O \in \mathcal A$ must have a form $ O =  h_{a_0} \otimes P_{a_0^c} +  (I-h_{a_0}) \otimes P'_{a_0^c} $
for some $c , c' \in \mathbb{C}$. But, the second term is annihilated by $h_{a_0}$.
If follows that $ O \equiv I_{a_0} \otimes P_{a_0^c} \in \mathcal A / \mathcal N $.
Hence, the logical algebra is invariant for adding or removing disentangled degrees of freedom $\ket 0$
together with Hamiltonian terms $\ket 0 \bra 0$.

\subsection{Example}
\label{sec:SmatrixExample}

Let us apply our general definition of the particle types to a simple model,
the $\mathbb{Z}_2$ gauge theory or the toric code model~\cite{Kitaev2003Fault-tolerant}.
This will help developing some intuition about $\mathcal A / \mathcal N$.
Our definition will reproduce what is known.
The model is defined on a two-dimensional square lattice
with two-dimensional degrees of freedom at each edge.
The Hamiltonian is
\[
H = - \sum_{\square} \frac{1}{2} \left(I+\prod_{i \in \square} \sigma^z_i \right) 
- \sum_{\mathlarger{\mathlarger +}} \frac{1}{2} \left( I + \prod_{ i \in {\mathlarger{ \mathlarger +}}} \sigma^x_i \right),
\]
where $\square$ denotes four edges associated to a plaquette and ${\mathlarger{ \mathlarger +}}$ denotes four edges associated to a vertex.
The identity operators are inserted to make $H$ conform with our assumption that the Hamiltonian
is a sum of projectors.
It is simple to verify that 
the product of $\sigma^z$ along any closed loop of edges on the real lattice
commutes with every vertex terms,
and the product of $\sigma^x$ along any closed loop of dual edges on the dual lattice
commutes with every plaquette terms.
Here, the closed loop does not have to consist of a single loop;
any product of the loop operators commute with every term of $H$.
Furthermore, any operator that commutes with every plaquette and vertex term
is a linear combination of the loop and dual loop operators.
Therefore, our algebra $\mathcal A$ is precisely generated 
by the loop and dual loop operators supported on a given annulus.

The set $\mathcal N$ of all null operators on the annulus has to be examined carefully.
Let $O$ be any null operator.
By definition, $PO=0$ where $P$ is the product of all $h_a$ whose support overlaps with the annulus.
Assume without loss of generality that $h_{a_1} \cdots h_{a_m} O = 0$ where $m$ is minimum possible.
We claim that each $h_{a_k}$ not only acts nontrivially on the annulus,
but is in fact supported on the annulus.
To see this, suppose this is not the case.
Then, one can choose a qubit $a$ outside the annulus
such that it is acted on by $h_{a_j}$ but not by $h_{a_k}$ for $k \neq j$.
Since $s_j=2h_{a_j}-I$ is a tensor product of Pauli matrices,
there is a Pauli matrix $\sigma$ acting on $a$ such that $\sigma s_j = - s_j \sigma$.
Using $\sigma h_{a_j} \sigma + h_{a_j} = I$,
we have
\[
0=\sum_{n=0}^1 \sigma^n \left( O \prod_{k=1}^m h_{a_k} \right) (\sigma^\dagger)^n 
 = O \prod_{k \neq j} h_{a_k},
\]
which is a contradiction to the minimality of $m$.
We conclude that
\[
 \mathcal N = \left\{ \sum_k O_k (1-h_k) ~\middle|~ O_k, h_k \in \mathcal A \right\}
\]
is a two-sided ideal of $\mathcal A$ generated by $(1-h_k) = (1- s_k)/2$ 
where the terms $h_k$ are supported on the annulus.

Next, we explain that, by factoring out by $\mathcal N$,
the operators of $\mathcal A$ can be ``deformed'' within the annulus.
If $O \in \mathcal A$ is a loop operator consisting of $\sigma^z$,
then $O s_a$ for some $\sigma^z$-type $s_a$ on the annulus that overlaps with $O$, is another loop
which is deformed from $O$. In the quotient algebra, the two operators are the same
\[
 O s_a \equiv O \in \mathcal A / \mathcal N
\]
because $O s_a - O = O(s_a -1) \in \mathcal N$.
Generalizing to multiple multiplications of elementary loops $s_a$,
we see that any two real or dual loop operators with the same winding number modulo 2
around the annulus are the same members of the quotient algebra.
Therefore, we have computed the logical algebra completely,
\[
 \mathcal A/ \mathcal N = \left\{ c_0 + c_1 \bar X + c_2 \bar Z + c_3 \bar X \bar Z  ~|~ c_i \in \mathbb{C} \right\},
\]
where $\bar X$ and $\bar Z$ are dual and real loop operators, respectively, that wrap around the annulus once.
The fundamental projectors of $\mathcal A/ \mathcal N$ are
\begin{align*}
 \pi_{1} & = \frac{1}{4} ( 1 + \bar X )( 1 + \bar Z ) \\
 \pi_{e} & = \frac{1}{4} ( 1 - \bar X )( 1 + \bar Z ) \\
 \pi_{m} & = \frac{1}{4} ( 1 + \bar X )( 1 - \bar Z ) \\
 \pi_{\epsilon} & = \frac{1}{4} ( 1 - \bar X )( 1 - \bar Z ) 
\end{align*}
We have named the projectors according to the conventional symbols for anyons of the model.
This verifies our identification of particle types with the fundamental projectors of $\mathcal A / \mathcal N$.
A similar calculation proves the same conclusion
for $\mathbb{Z}_N$ toric code, the quantum double model with the (gauge) group $\mathbb{Z}_N$.
The fundamental projectors are labeled by tuples $(a,b) \in \mathbb{Z}_N \times \mathbb{Z}_N$.

For general quantum double models, a similar calculation has been done by Bombin and Martin-Delgado% 
~\cite[App. B]{BombinMartin-Delgado2008Family},
who show that the algebra $\mathcal A$ on the thinnest possible annulus, known as a ribbon,
is spanned by projectors
onto definite particle type states.
This is consistent with our identification because $\mathcal N$ for the ribbon has to be zero according to their calculation.
However, the calculation is insufficient to verify our identification
since it does not show the structure of the operator algebra $\mathcal A$ or $\mathcal N$ for 
an annulus of non-minimal thickness.
Levin and Wen~\cite{LevinWen2005String-net} have computed (string) operators
that commute with every term of Hamiltonians possibly except at the ends
for the string-net models. They have noted that particle types are identified with
certain ``irreducible'' solutions of commutativity equations.
This irreducibility depends on the structure of their particular equations.
It would be a nontrivial task to adapt this irreducibility to our more general setting.

\subsection{Twist product and topological S-matrix}
\label{sec:tildeSmatrix}

Consider a pair of large circular annuli $A_L$ and $A_R$ as shown in Figure~\ref{fig:annuli}.
We require that the annuli are of comparable size with each other
and their center is separated by a distance comparable to the radius.
The intersection of the annuli then consists of two diamond-like regions $C_u$ and $C_d$
that are separated by a distance comparable to the radius.
The configuration is to ensure that any local operator can only intersect at most one of $C_u$ or $C_d$.
For each of the annuli, there is an algebra $\mathcal C^{(i)} = \mathcal A^{(i)} / \mathcal N^{(i)}$, where $i = L, R$.
We will see how to extract entries of topological $S$-matrix 
from $\mathcal C^{(L)}, \mathcal C^{(R)}$ by taking a special product.

\begin{defn}
For a pair of bipartite operators
\[
 P_{MM'} = \sum_{j} P_M^j \otimes P_{M'}^j, \quad
 Q_{MM'} = \sum_{k} Q_M^k \otimes Q_{M'}^k,
\]
the {\bf twist product} is defined by
\begin{equation}
 P_{MM'} \infty Q_{MM'} := \sum_{jk} P_M^j Q_M^k \otimes Q_{M'}^k P_{M'}^j .
 \label{eq:def-twist}
\end{equation}
\end{defn}
\noindent
That is, we reverse the order of the product in the subsystem $M'$.
Note that the twist product is independent of any linear decomposition of the operators
into tensor product operators. This is because Eq.~\eqref{eq:def-twist} is bilinear in its arguments:
\begin{align*}
 (aP_1 + b P_2) \infty Q &= a(P_1 \infty Q) + b(P_2 \infty Q), \\
 P \infty (aQ_1 + bQ_2) &= a(P \infty Q_1) + b(P \infty Q_2)
\end{align*}
for any $a,b \in \mathbb{C}$.
Unless specified otherwise, 
the operator $P$ that appears on the left of the symbol $\infty$ will always be supported on the left annulus,
and $Q$ on the right of $\infty$ will always be on the right annulus.

Let $M$ be a region that includes $C_u$ but not $C_d$, and let $M'$ be the complement of $M$.
If we take, with respect to this bipartition of the plane, 
the twist product of $P$ on the left annulus and $Q$ on the right annulus,
the resulting operator $ P \infty Q $ 
have indeed a twisted configuration compared to the usual product $PQ$.
See Figure~\ref{fig:annuli}.
Observe that the precise choice of $M$ is immaterial to $P \infty Q$.
$M$ can be the region above an arbitrary horizontal line placed between $C_u$ and $C_d$.
\begin{defn}
Given a state $\ket \psi$ and the pair of annuli,
\[
 (P,Q) \mapsto \bra \psi P \infty Q \ket \psi
\]
is a bilinear map from operators to complex numbers.
We will call it the {\bf twist pairing} of $P$ and $Q$.
\label{defn:twist-pairing}
\end{defn}

We particularly consider a matrix formed by the twist pairings of 
particle type projectors $\pi_a$ defined previously in Section~\ref{sec:particle-types}.
\begin{equation}
 \tilde S_{ab} := \bra \psi \pi_a^{(L)} \infty \pi_b^{(R)} \ket \psi . 
 \label{eq:tildeS}
\end{equation}
Here, we have abused notation.
$\pi_a \in \mathcal A / \mathcal N$ is an equivalence class of operators, 
and, strictly speaking, is not an operator.
The pairing should be read as the twist pairing of any representatives of the classes $\pi_a$ and $\pi_b$.
Although $\tilde S$ is defined by arbitrary representatives,
the resulting value is well-defined.
To see this, it is enough to show that
\begin{equation}
 \bra \psi N \infty O \ket \psi = 0
\label{eq:Hnull-null}
\end{equation}
for any operator $N \in \mathcal N_L$ on the left annulus, and $O \in \mathcal A_R$ on the right annulus.
Since both $N$ and $O$ are commuting with any term $h_j$ of the Hamiltonian,
we have
\[
 \bra \psi N \infty O \ket \psi = \bra \psi (N \infty O) h_j \ket \psi = \bra \psi (N h_j) \infty O \ket \psi .
\]
Applying this repeatedly, we can bring any number of $h_j$'s to $N$.
By definition, $N$ is an operator that is annihilated by some product of $h_j$'s.
Therefore, Eq.~\eqref{eq:Hnull-null} holds.

\subsection{Connection to modular tensor category}
\label{sec:connection}
Let us illustrate how $\tilde S$ is related to the topological S-matrix.
Our discussion will be heuristic and rather brief,
and we will not attempt to establish the relation from first principles,
as we do not claim that 
we have proved the correspondence of the particle types with the fundamental projectors.
Suppose that our system has excitations that are described by a modular tensor category,
and that the operators on the annulus can be unambiguously decomposed as
a linear combination of closed string operators carrying particles of definite type.
In the modular tensor category, a linking diagram represents morphisms
where vertical segments are labeled by objects.
We identify the lines of a linking diagram as worldlines of physical anyons.
For any particle type $a$,
the closed string operator for the particle type $a$ is normalized as
\[
 \bra \psi \bigcirc_a \ket \psi = d_a \ge 1 ,
\]
where $d_a$, the quantum dimension of $a$,
is determined by the fusion rules of the modular tensor category.
Then, the topological S-matrix is given by
\begin{equation}
 S_{ab} = \frac{1}{\mathcal D} \bra \psi \bigcirc_a \infty \bigcirc_b \ket \psi ^*,
 \label{eq:Smatrix}
\end{equation}
where the complex conjugation is to follow the traditional convention;
see \cite[Chap. 3]{BakalovKirillov2001},
\cite[App. E]{Kitaev2006Anyons}, and Figure~\ref{fig:annuli}.
The number $\mathcal D = \sqrt{ \sum_a d_a^2 }$ is the total quantum dimension.
For the distinguished label $a=$``1'', the vacuum, we have $\bigcirc_1 = I$.
It follows that
\[
S_{1a} = d_a / \mathcal D
\]
by the normalization.

In writing Eq.~\eqref{eq:Smatrix}, 
we have exploited an expected property of any topological theory
that the expectation value of any linked operators 
depends only on the topology of their configuration in spacetime.
As noted in the beginning of this section, an entry of the topological S-matrix is
an amplitude of a braiding process.
Hence, the operators inserted between $\bra \psi$ and $\ket \psi$
should have been a linking diagram that is extended in the time direction;
taking the vertical direction of this piece of paper as the time $t$
and the horizontal direction as one of the space directions, say $x$,
the diagram should look like Figure~\ref{fig:annuli}(c).
Relying on the expected property of topological theory,
we rotate the operators in spacetime so that the vertical direction
of this piece of paper is now $y$ direction of the space.
For these rotated operators our twisted product is well-suited.

The particle type projection operators $\pi_a$ as a linear combination of the string operators $\bigcirc_a$
can be obtained using Lemma~3.1.4 of Ref.~\cite{BakalovKirillov2001} or Eq.~(225) of Ref.~\cite{Kitaev2006Anyons}
\begin{equation}
\raisebox{-5ex}{\includegraphics[height=10ex]{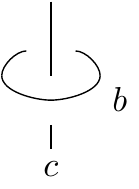}} = \frac{S_{bc}^*}{S_{1c}} 
\raisebox{-5ex}{\includegraphics[height=10ex]{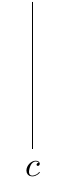}}
\end{equation}
and the unitarity of the topological $S$-matrix.
If $\pi_a = \sum_b \xi_{ab} \bigcirc_b$, then
\[
\delta_{ac} \raisebox{-5ex}{\includegraphics[height=10ex]{line-c.pdf}}
=
\raisebox{-5ex}{\includegraphics[height=10ex]{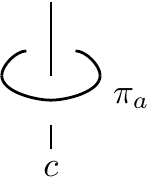}}
=
\sum_b \xi_{ab} \raisebox{-5ex}{\includegraphics[height=10ex]{link-cb.pdf}}
=
\sum_b \xi_{ab} \frac{S_{bc}^*}{S_{1c}} 
\raisebox{-5ex}{\includegraphics[height=10ex]{line-c.pdf}}
\]
Hence,
\begin{equation}
 \pi_a = \frac{d_a}{\mathcal D} \sum_b S_{ab} \bigcirc_b .
\end{equation}
Therefore, using the linearity of twist pairing, we have
\begin{equation}
 \tilde S_{ab} = \frac{d_a d_b}{\mathcal D}S_{ab}.
 \label{eq:tildeSisDSD}
\end{equation}

We can recover $S$ from $\tilde S$ since the trivial particle or vacuum is distinguished.
There exists a unique label $a=$``1'' such that $\pi_1 \ket \psi = \ket \psi$, 
and all other $\pi_a$ with $a \neq 1$ annihilates $\ket \psi$.
The uniqueness of ``1'' is not immediate from what we have assumed for constructing $\pi_a$,
but will follow from the local topological order condition that we will explain in 
Remark~\ref{rem:vacuum-distinguished} in Section~\ref{sec:ass-ltqo}.
On the other hand, the uniqueness is a part of assumptions of the modular tensor category.

The quantum dimensions of each particle type can be read as
\[
 \tilde S_{1a} = \frac{d_a^2}{\mathcal D^2} = \tilde S_{a1} .
\]
The total quantum dimension can be read off from $\tilde S_{11} = 1/ \mathcal D^2$.
Therefore, $\mathcal D$ and $d_a$ are determined, and $S$ is reconstructed from $\tilde S$.

Remark that except for the unique label ``1'' the rows and columns of 
the matrix $\tilde S$ have no preferred ordering.
In addition, without appealing to some homogeneity of the state,
there is no guarantee that the number of rows and that of columns are the same.
Indeed, if an annulus straddles two regions, one of a topological state and the other of a trivial state,
separated by a gapped boundary,
into which some of the particles can be absorbed~\cite{BravyiKitaev1998,KitaevKong2011Boundary},
the logical algebra can have smaller dimension as a vector space
than it did when there was no boundary.

We conclude this section
by proving that the value of the twist pairing is invariant under small-depth quantum circuit
in the following sense.
\begin{lem}
For any pair of operators $P$ on the left annulus and $Q$ on the right,
and for any local quantum circuit $W$ of range $R$ where 
\begin{equation}
R < \mathrm{dist}(C_u,C_d)/10
\label{eq:range-restriction}
\end{equation}
it holds that
\begin{equation}
 W( P \infty Q )W^\dagger = (W P W^\dagger) \infty ( W Q W^\dagger).
 \label{eq:inv-twist}
\end{equation}
\label{lem:inv-twist-prod}
\end{lem}
\begin{proof}
The range of the circuit is conditioned (Eq.~\eqref{eq:range-restriction}) 
such that even if one fattens the annuli by distance $R$,
the intersection of the fattened annuli will still consist of two separated regions.

Since the local quantum circuit is a product of local unitaries,
it suffices to show Eq.~\eqref{eq:inv-twist} for a local unitary $U$ acting on at most, say, 2 nearby qubits.
From the defining formula Eq.~\eqref{eq:def-twist}
our claim is clear if $U$ is acting on either $M$ or $M'$ but not on both.
If $U$ is acting across $M$ and $M'$, we can always choose a different $M$
such that $M$ includes or excludes the support of $U$.
The choice is possible because the intersection of the two operators $P,Q$ 
is well-separated.
\end{proof}

This lemma roughly suggests that the topological S-matrix is invariant under any quantum circuit
if the circuit is not too deep.
This is because the algebra on the annulus is defined in terms of the commutativity with terms of the Hamiltonian,
and the deformation of the state and the Hamiltonian by a quantum circuit leaves all commutation relation untouched. 
Hence, the quotient algebra $\mathcal A / \mathcal N$ is anticipated
to be isomorphic by a map induced by the quantum circuit conjugation.
Then, the fundamental projectors before and after the transformation will be corresponded unambiguously,
and their twist products, and hence the matrix $\tilde S$, are the same by the lemma.

This sounds reasonable, but we find it difficult to prove it rigorously
along the line we just described.
The first problem is that
the anticipation that the quotient algebra
$\mathcal A / \mathcal N$ would be invariant may not be true in general.
In order to deal with quantum circuits,
it appears that we have to consider situations where
the support of the logical algebra is enlarged.
In Remark~\ref{remark-nonstable-logical-algebra} 
on page~\pageref{remark-nonstable-logical-algebra},
we discuss a stack of two-dimensional toric code layers in three dimensions
and consider its logical algebra supported on the (thick) wall of a sphere.
The logical algebra changes as its support is enlarged.
This happens essentially because 
there are formally infinitely many particle types,
which might be the only reason.
The second problem is in our explicit use of the Hamiltonian 
to construct the algebra $\mathcal A/ \mathcal N$.
We wish to find an invariant of the state,
and this approach alone does not give any Hamiltonian-independent quantity.

Note that our objects and statements can be tested without invoking
concepts of anyons or framework of modular tensor categories,
though the former are certainly motivated by the latter.
In the next sections, we will assume an extra condition on the quotient algebra $\mathcal A / \mathcal N$,
and prove that indeed the $\tilde S$ is invariant under small-depth quantum circuits.
Furthermore, we will prove that there exists a certain class of Hamiltonians
such that whenever they share a ground state the corresponding $\tilde S$ matrices are the same.
This supports a belief that a ground state wave function contains 
all information about the phase of matter it represents.
The independence of the $\tilde S$ on the Hamiltonian
leads us to a lower bound on the depth of any transformation quantum circuits
between states with distinct $\tilde S$ matrices.

\section{Locally invisible operators: Many-body entanglement witness}
\label{sec:locally-invisible}

In this section,
we define a class of operators, which we call {\em locally invisible} operators,
that manifestly depends only on a given state,
and study their properties.
We will show that the existence of a nontrivial pair of locally invisible operators
is a sufficient condition that the state requires a deep local quantum circuit to be generated.
In general, it is a very hard problem to give a lower bound on the complexity of the quantum circuit.
In some physically important cases,
it is sufficient to examine space-correlation functions to give such a lower bound,
since quantum circuits of small depth can only generate short-ranged correlations.
The problem is to determine what correlation functions to examine.
Our approach may be considered as a way to find those correlation functions.

The sufficient condition given by the locally invisible operators
defines a ``many-body entanglement witness.''
To the best of the author's knowledge,
this is the first rigorous witness 
that only depends on the bulk of the wave-function
and that is applicable to topologically ordered states.

To motivate the notion,
let us recall what was physically important for elements of our algebra.
The string operators are responsible 
for transporting anyons through the system,
and their algebraic relations determine all topological data of the state.
The precise location of the string is never important,
but what matters is the topology of the strings
whether the string encircles a particular region
or whether the string has end points.
This demands that the strings be locally invisible;
any local observable on the support of the string away from the end points
should not reveal whether the string operator has acted on the ground state.
In exactly solved examples, this invisibility follows from the fact that
there are many equivalent string operators.

We characterize the invisibility without resorting to the Hamiltonian as follows.
Let $\ket \psi$ be an arbitrary quantum state on some lattice,
and $A \subseteq B$ be two regions.
\begin{defn}
An operator $O$ is said to be {\bf locally invisible at $(A,B)$ with respect to $\ket \psi$}
if for any $\ket \phi$ such that
\[
 \mathrm{tr}_{B^c} \left[ \ket \phi \bra \phi \right] = \mathrm{tr}_{B^c}\left[ \ket \psi \bra \psi \right],
\]
one has
\[
 \mathrm{tr}_{A^c} \left[ O \ket \phi \bra \phi O^\dagger \right] \propto \mathrm{tr}_{A^c} \left[ \ket \psi \bra \psi \right].
\]
If $O$ is locally invisible on $\ket \psi$ at every pair $(A,B)$
where $A$ is any disk of radius $r$, and $B$ is $t$-ball of $A$,
then we simply say that $O$ is {\bf $(r,t)$-locally invisible with respect to $\ket \psi$}.
\end{defn}

A locally invisible operator determines the local reduced density matrix on $A$
based on the local reduced density matrix on $B$ as if it were the identity operator.
In particular, it does not change the local reduced density matrix when acted on $\ket \psi$.
A direct consequence is that any locally invisible operator on a product (trivial) state
actually acts like a scalar multiplication on the state.
We will provide examples shortly.

\begin{lem}
Suppose $O$ is $(r,t)$-locally invisible on $\ket \psi = \ket{00 \cdots 0}$ for some $r \ge 1$.
Then, $O \ket \psi = \ket \psi \bra \psi O \ket \psi$.
\label{lem:local-id-prod-stabilizer}
\end{lem}
\begin{proof}
The product state is defined by local projectors $\Pi_i = \ket 0 \bra 0 _i$ 
acting on site $i$
by the equations $\Pi_i \ket \psi = \ket \psi$.
The local invisibility implies $\Pi_i (O \ket \psi) = (O \ket \psi$).
This equation actually determines the state $O \ket \psi$ up to scalars.
\end{proof}

On a general state $\ket \psi$, however, the locally invisible operator
need not be a stabilizer for $\ket \psi$.
As we will see later, it might map $\ket \psi$ to an orthogonal state or some other state.
By the lemma, this happens only if $\ket \psi$ is not a product state.
In fact, the existence of any locally invisible operator with nontrivial action
indicates that $\ket \psi$ is not even close to the product state.
To see this, let us first show that the local invisibility 
is invariant under local unitary transformations.

\begin{lem}
Let $W$ be a local quantum circuit of range $R$.
If $O$ is $(r,t)$-locally invisible with respect to $\ket \psi$, where $r > R$,
then $WOW^\dagger$ is $(r-R,t+R)$-locally invisible with respect to $W \ket \psi$.
\label{lem:evolved-invisibility}
\end{lem}

This easily follows from the general observation
that for any state $\rho$ (mixed or pure)
the reduced density matrix of $W \rho W^\dagger$ on a ball of radius $r$
is determined by the reduced density matrix of $\rho$ on the concentric ball of radius $r+R$
whenever the range of the quantum circuit $W$ is $R$.

\begin{proof}
Let $A', A, B, B'$ be concentric disks of radius $r-R$, $r$, $r+t$, $r+t+R$, respectively.
Let $\ket \phi$ be a state with the same reduced density matrix on $B'$ as $W \ket \psi$.
We have to show that $WOW^\dagger \ket \phi$
has the same reduced density matrix on $A'$ as $W \ket \psi$.

The state $W^\dagger \ket \phi$ has the reduced density matrix on $B$ that is the same as that of $\ket \psi$.
By the local invisibility, the state $O W^\dagger \ket \phi$ has the same reduced density matrix on $A$
as $\ket \psi$.
This implies that $WOW^\dagger \ket \phi$ has the reduced density matrix on $A'$ that is the same as that of $W \ket \psi$.
\end{proof}
Thus, given a locally invisible operator on a state $\ket \psi$,
one knows at least one locally invisible operator on a ``perturbed'' state $W \ket \psi$.

\subsection{Trivial states}

Here, we compute all possible values of twist pairing of locally invisible operators with respect to a product state.
\begin{theorem}
Let $P$ and $Q$ be $(r,t)$-locally invisible operators on two annuli
with respect to a product state $\ket \psi$,
where $r \ge 1$ and $t \ge 0$ such that the two regions of intersection 
of the annuli are separated by a distance $> 2(r+t)$.
Then,
\[
 \bra \psi P \infty Q \ket \psi = \bra \psi P \ket \psi \bra \psi Q \ket \psi.
\]
\label{thm:product-state}
\end{theorem}
\begin{proof}
By Lemma~\ref{lem:inv-twist-prod}, we can assume $\ket \psi = \ket{00\cdots0}$.
The reduced density operator of $\ket \psi$ is a rank 1 projector 
$\Pi_M = \ket{0\cdots 0}\bra{0\cdots 0}$ for any region $M$.
Let $A, B$ be concentric disks of radius $r$ and $r+t$, respectively.
Let $\ket \phi$ be a state such
that the reduced density matrix on $B$ is the same as that of $\ket \psi$.
This is equivalent to the condition $\ket \phi = \Pi_B \ket{ \phi'}$ for some unnormalized $\ket{\phi'}$.
The local invisibility of $Q$ says $\Pi_A Q \ket \phi = Q \ket \phi$.
Thus, we have $\Pi_A Q \Pi_B \ket{\phi'} = Q \Pi_B \ket{\phi'}$, which holds for any $\ket{\phi'}$.
Therefore,
\begin{equation}
 \Pi_A Q \Pi_B = Q \Pi_B.
 \label{eq:local-identity}
\end{equation}
Note that the support of $\Pi_B$ overlaps with at most one region of intersection of the annuli.

Place $A$ and $B$ around the upper intersection region $C_u$; see Figure~\ref{fig:annuli}.
By Eq.~\eqref{eq:local-identity},
\[
 (P \infty Q) \ket \psi = (P \infty Q) \Pi_B \ket \psi = (( P \Pi_A ) \infty Q ) \Pi_B \ket \psi
\]
For other locations of $A$ and $B$, it obviously holds that
$(P \infty Q) \ket \psi = ((P \Pi_A) \infty Q) \ket \psi$.
We conclude that
\[
 (P \infty Q) \ket \psi = (( P \prod_A \Pi_A ) \infty Q ) \ket \psi
\]
where the product over $A$ covers the whole system.
But, $\prod_A \Pi_A = \ket \psi \bra \psi$ since $\ket \psi$ is a product state.
By Lemma~\ref{lem:local-id-prod-stabilizer}, $P \ket \psi = (\bra \psi P \ket \psi) \ket \psi$.
Our theorem is thus proved.
\end{proof}
\noindent
Note that Lemma~\ref{lem:local-id-prod-stabilizer} is essentially used.
Eq.~\eqref{eq:local-identity} may be true even if $\ket \psi$ is not a product state.

\subsection{Generating a state that admits nontrivial twist pairing}
As a trivial consequence of Theorem~\ref{thm:product-state},
we obtain a many-body entanglement witness.
\begin{cor}
For a pair of annuli of which two regions of the intersection
are separated by a distance $> 2(r+t)$,
suppose there are operators $P$ and $Q$ that are $(r,t)$-locally invisible with respect to $\ket \psi$,
where $r \ge 2$ and $t \ge 0$.
If 
\[
 \bra \psi P \infty Q \ket \psi \neq 
 \bra \psi P \ket \psi \bra \psi Q \ket \psi,
\]
then the range of any local quantum circuit $W$ 
generating $\ket \psi = W \ket {00\cdots 0}$ from a product state $\ket {00\cdots 0}$
must be at least $r/10$.
\label{cor:witness}
\end{cor}
\begin{proof}
Let $R$ be the range of $W$. Suppose $R < r/10$.
$W^\dagger P W$ and $W^\dagger Q W$ are $(r-R,t+R)$-locally invisible with respect to $W^\dagger \ket \psi$
by Lemma~\ref{lem:evolved-invisibility},
and their twist pairing is still the same by Lemma~\ref{lem:inv-twist-prod}.
This contradicts to Theorem~\ref{thm:product-state}.
\end{proof}

As a simple example, consider a trivial product state on a sphere.
Remove a qubit at the North pole and another at the South pole,
and bring a Bell pair $(\ket{00} + \ket{11})/\sqrt 2$ 
to place a half at the North pole and the other half at the South pole.
Now, let overall state be $\ket \psi$.
The state is stabilized by two-qubit non-local operators
\begin{align*}
 \bar X = \sigma^x_\text{South pole} \otimes \sigma^x_\text{North pole},\\
 \bar Z = \sigma^z_\text{South pole} \otimes \sigma^z_\text{North pole}.
\end{align*}
We regard one great circle through the poles as the left annulus,
and another great circle through the poles that is perpendicular to the first one as the right annulus.
The expectation values of the two non-local operators are $1$, but their twist pairing is $-1$.
The reduced density matrix at a pole consists of product state part and $I/2$ (completely mixed) for the pole.
Both $\bar X$ and $\bar Z$ are easily seen to be locally invisible.
\emph{In fact, any small-depth quantum circuit that stabilizes the state is locally invisible.}
(See the proof of Lemma~\ref{lem:evolved-invisibility}.)
Corollary~\ref{cor:witness} implies that
any quantum circuit that generates $\ket \psi$ from a product state must have depth 
that is at least linear in the system's diameter.
This linear bound is tight up to constants;
create a Bell pair locally, and transport a half of the Bell pair by swap operators.
This generating circuit for $\ket \psi$ 
has depth approximately equal to the circumference of the great circle.

A similar conclusion can be drawn for a topologically ordered state.
Consider a toric code state on the sphere.
We know that there are loop operators $\bar X$ and $\bar Z$ consisting of $\sigma^x$ and $\sigma^z$, respectively,
going around the great circles.
They both stabilizes the state $\bar X \ket \psi = \ket \psi = \bar Z \ket \psi$, but  their twist pairing is
\[
 \bra \psi \bar X \infty \bar Z \ket \psi = -1 .
\]
They are depth-1 quantum circuits that stabilize the state, so they are locally invisible.
Therefore, a linear-depth circuit is needed to generate the toric code state from a product state.
Note that since there is a unique state
the arguments of Refs.~\cite{BravyiHastingsVerstraete2006generation,KoenigPastawski2014Generating}
that utilize the local indistinguishability between two orthogonal states
cannot be applied.

Before we end this section,
we mention a marginal generalization of Theorem~\ref{thm:product-state}.
Though we have exclusively stated the theorem using the annuli,
the geometry of the annuli is not too important.
The key property of the two annuli is that the intersection consists of far separated two regions,
where two locally invisible operators are respectively supported.
The theorem remains true in a situation where
$P$ acts on the whole system and $Q$ acts on a far separated two disks.
The proof for this situation requires hardly any modification.
This generalization can be applied to
the Greenberger-Horne-Zeilinger state
\[
\ket{\text{GHZ}} = \frac{1}{\sqrt{2}} \ket{00\cdots 0} + \frac{1}{\sqrt{2}}\ket{11\cdots 1}
\]
on a line or on a plane.
The GHZ state is stabilized by $\bar X = (\sigma^x)^{\otimes n}$ 
and $\bar Z = \sigma^z_i \otimes \sigma^z_j$,
where $n$ is the total number of qubits, and $i,j$ are distinct qubits.
Taking the qubits $i$ and $j$ that are far apart,
we see that the twist pairing of $\bar X$ and $\bar Z$ is $-1$.
Therefore, the GHZ state requires a linear depth quantum circuit to generate.

\section{An invariant of states under shallow quantum circuits}
\label{sec:invariant}

The locally invisible operators discussed in the previous section
is defined in terms of a given state. However, they do not in general form an algebra;
if $P$ and $Q$ are $(r,t)$-locally invisible operators, neither $P+Q$ nor $PQ$ has
to be $(r,t)$-locally invisible.
In order to choose useful and canonical elements from the set of all locally invisible operators with respect to a ground state,
we resort to Hamiltonians, and effectively give an algebra structure to the set of all locally invisible operators.

To this end, we assume technical conditions on the Hamiltonian.
The first one is so-called {\em local topological order} condition, which, 
roughly speaking, asserts that
local reduced density matrix of a ground state is determined locally.
This condition is used in a gap stability proof of topological order~\cite{BravyiHastings2011short,MichalakisZwolak2013Stability}.
The second one is an independence of the logical algebra $\mathcal A / \mathcal N$ on the size of its support.
Recall that the logical algebra is the algebra $\mathcal A$ of all (string) operators on an annulus that commute
with every term of the Hamiltonian,
modulo null operators $\mathcal N$ that annihilate any state without excitations on the annulus.
(See Section~\ref{sec:tildeSmatrix} or 
Definition~\ref{defn:H-null} in Section~\ref{sec:ass-stable-logical-algebra} below.)
The algebra $\mathcal A$ and its ideal $\mathcal N$ generally depends on the radius of the annulus
and the thickness of the radius.
The second condition, which we call the \emph{stable logical algebra condition}, asserts
that the logical algebra is independent of the thickness for a given radius.
We will rigorously state the conditions in the following subsections.

\begin{theorem}
Suppose a state $\ket \psi$ on a plane of radius $>L$
admits an unfrustrated local commuting projector Hamiltonian $H$ of interaction range $w$
satisfying the local topological order condition and the stable logical algebra condition
such that $\ket \psi$ is a ground state of $H$.
Then, whenever $ 1200 w < 60 t < \rann < L $ the following is true.

For any such Hamiltonians $H_1, H_2$, the logical algebras
\[
 \mathcal C_{i} = \mathcal A_t^{H_i} / \mathcal N_t^{H_i} \quad i = 1,2
\]
on an annulus of radius $\rann$ and thickness $t$
are isomorphic by a map that preserves twist pairing.
Thus, $\mathcal C(\ket \psi) = \mathcal C_i$ is well-defined in terms of $\ket \psi$.
Moreover, for any quantum circuit $W$ of range $< t$,
the logical algebra $\mathcal C( W \ket \psi )$ is isomorphic to $\mathcal C(\ket \psi)$.
In particular, the $\tilde S$-matrix for $\ket \psi$ defined in Eq.~\eqref{eq:tildeS} is invariant under $W$ 
up to permutations of rows and columns.
\label{thm:invariant}
\end{theorem}
\begin{proof}[Sketch of the proof]
We will establish the following statements in Section~\ref{sec:proof}.
\begin{itemize}
\item[(i)] The local topological order condition implies that any operator of $\mathcal A_t$ is locally invisible.
Conversely, any locally invisible operator can be ``symmetrized'' to become a member of $\mathcal A_{t'}$,
where $t'$ is slightly larger than $t$.
\item[(ii)] The stable logical algebra condition implies that the logical algebra $\mathcal A_t / \mathcal N_t$
 is in bijection with the set of all locally invisible operators modulo ``locally null'' operators.
\item[(iii)] Since the locally invisible and null operators are defined in terms of the state, 
the logical algebra is the same regardless of the Hamiltonian whenever (i) and (ii) are valid. This proves the first part.
\item[(iv)] The evolved Hamiltonian $WHW^\dagger$ satisfies the local topological order and stable logical algebra conditions.
Therefore, the logical algebra for $W \ket \psi$ can be calculated from $WHW^\dagger$,
and is thus isomorphic to the logical algebra for $\ket \psi$ by the conjugation $W \cdot W^\dagger$.
\item[(v)] The set of fundamental projectors is uniquely determined by the algebra. 
We know it for $\mathcal C^{WHW^\dagger}$ from that of $\mathcal C^{H}$.
The twist product is not affected by small depth quantum circuits.
This proves the second part.
\end{itemize}
\end{proof}

\subsection{Transformations between ground states of quantum double models with finite abelian groups}

From the calculations of Section~\ref{sec:SmatrixExample}, we know that the 
$\mathbb{Z}_N$-toric code Hamiltonian is unfrustrated and commuting, 
and satisfies the stable logical algebra condition.
By a criterion in Ref.~\cite{BravyiHastingsMichalakis2010stability} 
we know that it also satisfies the local topological order condition.
Juxtaposing a layer of $\mathbb{Z}_N$-toric code with another layer of $\mathbb{Z}_{N'}$-toric code,
we have a ground state of a quantum double model with the group $\mathbb{Z}_N \times \mathbb{Z}_{N'}$~\cite{Kitaev2003Fault-tolerant}.
This is because, firstly, a gauge transformation under a product group
is a product of gauge transformations of the component groups,
and, secondly, a vanishing flux condition for a product group 
is equivalent to vanishing flux conditions for each component group.
Since any finite abelian group is a direct product of finitely many cyclic groups,
we see that the ground state of a quantum double model with a finite abelian group
satisfies the conditions of Theorem~\ref{thm:invariant}.
Inserting any ancillary qudits in the trivial product state does not change the logical algebra at all.
Therefore,
we have a separation of phases of matter by quantum circuits of at least linear depth:
\begin{theorem}
Let $\ket{\psi(G)}$ denote a ground state of quantum double model with any finite abelian group $G$
with possible ancillary qudits in the trivial product state.
Let $L$ be the radius of a disk contained in the system.
If there is a local quantum circuit $W$ such that $\ket{\psi(G')} = W \ket{\psi(G)}$, then
either $G \cong G'$ or the depth of $W$ must be at least $cL$ for some constant $c>0$.
\label{thm:separation}
\end{theorem}

In view of a modular tensor category description of anyons,
the theorem is a simple corollary of the Verlinde formula
since Theorem~\ref{thm:invariant} guarantees that the $S$-matrices
of $\ket{\psi(G')}$ and $W \ket{\psi(G)}$ are the same.
For an abelian group $G$, the anyons of the quantum double model are all abelian,
and the fusion rules are the same as the group operation of $D(G) \cong G \times G$.
The Verlinde formula~\cite{Verlinde1988}
\[
 N^c_{ab} = \sum_x \frac{S_{ax}S_{bx}S_{cx}^*}{S_{1x}}
\]
tells us how to extract the fusion rules out of the $S$-matrix.
Therefore, the group $G \times G$ and hence $G$ is reconstructed.
The proof below is a direct translation of this argument into what we have rigorously.

\begin{proof}
Without loss of generality, we can assume that the whole system is on a sphere of radius $L$.
Theorem~\ref{thm:invariant} says that if $W$ had small depth,
then $\tilde S( \ket{\psi(G)} ) = \tilde S( W \ket{\psi(G)} ) = \tilde S( \ket{ \psi(G')})$.
(The independence of $\tilde S$ on Hamiltonians is essentially used.)
We will reconstruct the group $G$ from $\tilde S ( \ket{ \psi(G)} )$.

As we have just noted above, if $G = H \times K$ is a product group,
then $\ket{\psi(G)} = \ket{\psi(H)} \otimes \ket{\psi(K)}$.
Thus, the logical algebra $\mathcal C(G)$ on a (left) annulus of $\ket{\psi(G)}$ 
is a tensor product of those of $\ket{\psi(H)}$ and $\ket{\psi(K)}$:
\[
 \mathcal C(G) = \mathcal C(H) \otimes \mathcal C(K) .
\]
Then, the fundamental projectors of the logical algebra that give the decomposition into simple subalgebras
are the tensor products of those of $\mathcal C(H)$ and $\mathcal C(K)$:
\[
 \pi^G_{(ab)} = \pi^H_a \otimes \pi^K_b .
\]
where $(ab)$ is a double index. It follows that
\begin{align}
 \tilde S_{(ab)(a'b')}(G) &= \bra{\psi(G)} \pi^G_{(ab)} \infty \pi^G_{(a'b')} \ket{\psi(G)} \nonumber \\
 &= \bra{\psi(H)} \pi^H_a \infty \pi^H_{a'} \ket{\psi(H)} 
 \cdot \bra{\psi(K)} \pi^K_b \infty \pi^K_{b'} \ket{\psi(K)} \nonumber \\
 &= \tilde S_{aa'}(H) \tilde S_{bb'}(K) \label{eq:tildeS-tensorproduct}.
\end{align}
The matrix $\tilde S^{(d)}(\ket{\psi(G)})$ for a cyclic group $G = \mathbb Z_d$ is simple to compute.
The fundamental projectors are
\[
 \pi_{a_x a_z} = \frac{1}{d^2}\left( \sum_{k=0}^{d-1} \omega_d ^{k a_x} \bar X^k \right) \left( \sum_{k=0}^{d-1} \omega_d^{k a_z} \bar Z ^k \right)
\]
where $\omega_d = \exp( 2\pi i / d)$ is the $d$-th root of unity, 
$a_x, a_z = 0,1,\ldots,d-1$ label the $d^2$ fundamental projectors,
and $\bar X, \bar Z$ are loop operators consisting respectively of
unitary operators
\begin{equation}
X = \sum_{k \in \mathbb Z_N} \ket{ k + 1 \mod N}\bra{k},
%X = \begin{pmatrix}
%0&  &  &   & 1  \\
%1&0&  &  &     \\
%  & 1&0& & \\
%  &  & \ddots & \ddots & \\
%  &  &             &1& 0 
%\end{pmatrix},
\quad
Z = \sum_{k \in \mathbb Z_N} \omega_d^k \ket{k} \bra{k},
%Z = \begin{pmatrix}
%1 & & & & \\
%& \omega_d & & & \\
%& & \omega_d^2 & & \\
%& & & \ddots & \\
%&&&& \omega_d^{N-1}
%\end{pmatrix},
\label{eq:XZ}
\end{equation}
and their inverses. The loop operators satisfy
\begin{align*}
\bar X^n \bar Z^m \infty \bar X^{n'} \bar Z^{m'} = \omega_d^{-nm'-n'm} \bar X^{n+n'} \bar Z^{m+m'},
\end{align*}
We have used our convention that an operator on the left (right) of the symbol $\infty$ is on the left (right) annulus.
(The phase factor $\omega_d^{-nm'-n'm}$ could be inversed, depending on the orientation convention of the lattice.)
A direct computation yields
\[
 \tilde S^{(d)}_{(a_x a_z),( a'_x a'_z)} = \frac{1}{d^2} \omega_d^{a_z a'_x + a_x a'_z}.
\]
(This is consistent with Eq.~\eqref{eq:tildeSisDSD}.)
Consider
\begin{equation}
 N^c_{ab} := d^4 \sum_p \tilde S^{(d)}_{ap} \tilde S^{(d)}_{bp} \tilde S^{(d)*}_{c p}
 = \delta_{a_x + b_x, c_x} \delta_{a_z + b_z, c_z}
 \label{eq:VerlindeFormula}
\end{equation}
where $\delta$ is the Kronecker delta and $+$ in the subscript of $\delta$ is modulo $d$.
This means that $N^c_{ab}$ reveals the group structure of $\mathbb{Z}_d \times \mathbb{Z}_d$
on the unsorted label set $\{ a \} = \{ (a_x, a_z) : a_x, a_z \in \mathbb{Z}_d \}$ of the fundamental projectors.

Eq.~\eqref{eq:VerlindeFormula} generalizes to any product group $G$ of cyclic groups
using Eq.~\eqref{eq:tildeS-tensorproduct}.
Therefore, the group $G \times G$ can be reconstructed just from $\tilde S$ whose rows and columns are unsorted.
The finite abelian group $G \times G$ is uniquely determined by $G$,
which completes the proof.
\end{proof}

We believe that the group $G$ being abelian is not essential in the theorem.
For a nonabelian group $G$, one has to check
the stable logical algebra condition, which would involve 
nontrivial computations~\cite{BombinMartin-Delgado2008Family}.
The local topological order condition is proved, albeit implicitly, in~\cite[Sec.~3]{FiedlerNaaijkens2014Haag}.
Note that Theorem~\ref{thm:invariant} states that the $\tilde S$-matrix is invariant,
and we have shown for the abelian quantum double models $\tilde S$ is a complete invariant.
We have not addressed how fine or coarse the invariant $\tilde S$ is in general.
However, it will \emph{not} be the case that each quantum double model with general finite group $G$
gives a unique $\tilde S$. This is because there exist non-isomorphic groups 
that give rise to isomorphic modular tensor categories~\cite{Naidu2007}.

\subsection{Assumption I: Local topological order}
\label{sec:ass-ltqo}

In this and the next subsections, we define the assumptions of Theorem~\ref{thm:invariant}
and explain how they are physically motivated.

Suppose $\ket \psi$ is any ground state of a local commuting Hamiltonian
\begin{equation}
H = - \sum_j h_j ,
\label{eq:local-H}
\end{equation}
where $h_j=h_j^2$ are local projectors such that $[h_j,h_k] = 0$.
We assume that the Hamiltonian is frustration-free:
\begin{equation}
 h_j \ket \psi = \ket \psi \text{ for all } j.
 \label{eq:stabilizer}
\end{equation}
Further we assume that the reduced density matrix on a disk $D$ is determined by those $h_j$
that acts nontrivially on $D$ whenever $D$ does not wrap around the whole system.

That is, if there is a state $\ket \phi$ such that
\begin{equation}
 h_j \ket \phi = \ket \phi \text{ for every $h_j$ that meets the disk $D$},
 \label{eq:local-ground-state}
\end{equation}
then the reduced density matrix of $\ket \phi$ on $D$ is the same as that of $\ket \psi$
\begin{equation}
 \mathrm{tr}_{D^c} \ket \phi \bra \phi = \mathrm{tr}_{D^c} \ket \psi \bra \psi .
 \label{eq:same-rho}
\end{equation}
This condition is called {\bf local topological order} condition~\cite{MichalakisZwolak2013Stability},
and has been used to show the stability of the energy gap under small but arbitrary perturbations~\cite{BravyiHastingsMichalakis2010stability}.
This is a property of the Hamiltonian; there could be two different Hamiltonians whose ground states are equal
while one satisfies the local topological order condition but the other does not.
It is somewhat related to the homogeneity of the Hamiltonian.
Consider $H = - \sum_i \sigma^z_i \sigma^z_{i+1} - \sigma^z_0$,
the Ising model on a line with magnetic field at one particular spin.
The ground state is unique with all spins up.
On a region far from the spin where the field term acts, all-spin-down state satisfies Eq.~\eqref{eq:local-ground-state}.
Thus, Eq.~\eqref{eq:same-rho} cannot be satisfied.
On the other hand, $H_0 = -\sum_i \sigma^z_i$ has the same unique ground state
such that Eq.~\eqref{eq:local-ground-state} implies Eq.~\eqref{eq:same-rho}.

We allow the Hamiltonian to have degenerate ground space, though we do not require it.
If $H$ has degeneracy, which might depend on the topology of the underlying lattice,
then our local topological order condition requires that they should be locally indistinguishable.
Therefore, any classical Hamiltonian with local topological order condition
is essentially the trivial Hamiltonian $H_0$ with non-degenerate ground state.
Although the trivial product state is not usually termed as topologically ordered,
this is consistent since the trivial state can be regarded as an anyon system with the sole anyon, the vacuum.

\begin{remark}\label{rem:vacuum-distinguished}
In Section~\ref{sec:connection} we have mentioned that the vacuum label ``1'' is distinguished.
This is a consequence of the local topological order condition.
Recall that the particle types are defined by the fundamental projectors of the logical algebra 
$\mathcal C = \mathcal A / \mathcal N$ on an annulus of some large size.
An operator of the logical algebra is not necessarily local, but is anyway supported on
a disk that does not wrap the whole system.
By definition of the logical algebra, $O \ket \psi$ is a ground state for any $O \in \mathcal C$.
(We are abusing notation here as before. $O\ket \psi$ should be understood as
the state $O' \ket \psi$ for any representative $O' \in O \in \mathcal A / \mathcal N$.
It is well-defined since $\mathcal N$ annihilates any ground state.
See the discussion below Eq.~\eqref{eq:tildeS}.)
By the local topological order condition, we must have $O \ket \psi = c(O) \ket \psi$
for some number $c(O)$.
If $\pi_a$ are fundamental projectors,
then $\sum_a \pi_a \ket \psi = \ket \psi$ so $\sum_a c(\pi_a) = 1$,
but $\pi_a \pi_b \ket \psi = c(\pi_a)c(\pi_b) \ket \psi = 0$ whenever $a \neq b$.
See Eq.~\eqref{eq:fundamentalprojectors}.
Therefore, one and only one $c(\pi_a)$ is nonzero,
and the corresponding label $a =$ ``1'' represents the distinguished vacuum.
\hfill $\lozenge$
\end{remark}

Kitaev's quantum double model Hamiltonians~\cite{Kitaev2003Fault-tolerant}
satisfy the local topological order condition.
Ref.~\cite[Sec.~3]{FiedlerNaaijkens2014Haag} does not explicitly mention this,
but it is proved that the expectation value of local observables are determined
by Hamiltonian terms that acts nearby,
which is equivalent to the present local topological order condition.
We believe that Levin and Wen's string-net model Hamiltonians~\cite{LevinWen2005String-net}
also satisfy the local topological order condition for the following reason.
The ground state wave function can be viewed as a string condensate
--- a superposition of certain loop configurations.
The local terms in the Hamiltonians are just enough to infer
allowed configurations of the loops, their local deformation rules,
and the relative amplitudes among the deformed loops,
which are enough to determine the local reduced density matrix.

The following is a simple fact implied by the local topological order condition.
We will use it frequently in the proof of Theorem~\ref{thm:invariant}.
\begin{lem}
Let $D$ be a disk and $\Pi_D$ be the projector onto the subspace support of the reduced density operator
on $D$ of the ground state $\ket \psi$. If $h_j$ is a term of the Hamiltonian Eq.~\eqref{eq:local-H} whose support is
contained in $D$, then
\[
 h_j \Pi_D = \Pi_D .
\]
If $P_D$ is the product of all terms of  the Hamiltonian whose support overlap with $D$,
then
\[
 \Pi_D P_D = P_D.
\]
\label{lem:basic-projectors}
\end{lem}
\begin{proof}
The ground state $\ket \psi$ has the Schmidt decomposition
$
 \ket \psi = \sum_a \lambda_a \ket{\psi_D^a} \ket{\psi_{D^c}^a}  
$
where $\lambda_a > 0$. The projector is
$
 \Pi_D = \sum_a \ket {\psi_D^a}\bra{\psi_D^a}.
$
Since the Hamiltonian is frustration-free, we know
$h_j \ket \psi = \ket \psi$. In terms of the Schmidt decomposition,
\[
\sum_a \lambda_a h_j \ket{\psi_D^a} \ket{\psi_{D^c}^a}  = \sum_a \lambda_a \ket{\psi_D^a} \ket{\psi_{D^c}^a}.
\]
Since $\ket{\psi_{D^c}^a}$ are orthonormal, we must have
\[
 h_j \ket{\psi_D^a} = \ket{\psi_D^a}
\]
for each $a$. Therefore, $h_j \Pi_D = \Pi_D$.
The second assertion can be proved by the local topological order condition.
For any vector $\ket \phi$, $P_D \ket \phi$ has (after normalization)
the reduced density operator on $D$ which is equal to that of $\ket \psi$.
In other words, the Schmidt decomposition of $P_D \ket \phi$ must be
\[
 P_D \ket \phi \propto \sum_a \lambda_a \ket{\psi_D^a} \ket{\phi_{D^c}^a}
\]
for some orthonormal $\ket{\phi_{D^c}^a}$.
The equation $\Pi_D P_D = P_D$ follows since $\ket \phi$ was arbitrary.
\end{proof}

Technically, the conditions of Ref.~\cite{BravyiHastings2011short} are phrased slightly differently than ours.
It turns out that they are equivalent.
``TQO-2''~\cite{BravyiHastings2011short} states that
if $O_D$ is an operator supported on a disk $D$ such that $
O_D \ket \psi = 0 $
for any ground state $\ket \psi$, then $O_D P_D = 0$ where $P_D$ is as in Lemma~\ref{lem:basic-projectors}.
This readily follows from our definition.
If $O_D \ket \psi = 0$, then $O_D$ annihilates any Schmidt component on $D$ of $\ket \psi$,
which means $O_D \Pi_D = 0$, but then $O_D P_D = O_D \Pi_D P_D = 0$ by Lemma~\ref{lem:basic-projectors}.
The converse is contained in Corollary~1 of \cite{BravyiHastings2011short}.

\subsection{Assumption II: Stable logical algebras}
\label{sec:ass-stable-logical-algebra}

\begin{defn}
Given a Hamiltonian Eq.~\eqref{eq:local-H} of interaction range $w$,
an operator $O$ supported on an annulus $A_t$ of thickness $t$ is said to be {\bf $H$-null}
if it can be expressed as
\[
 O = O_1 + \cdots + O_n
\]
where
for each $O_i$ there exists a term $h_{j(i)}$ of $H$ on the annulus $A_{t+w}$ of thickness $t+w$ such that $h_{j(i)} O_i = O_i h_{j(i)} = 0$.
In addition, we define two sets:
\begin{itemize}
 \item $\mathcal A_t$: the set of all operators on the annulus of thickness $t$ that commute with every term of $H$. It is a $C^*$-algebra.
 \item $\mathcal N_t$: the subset of $\mathcal A_t$ consisting of all $H$-null operators. It is a two-sided ideal of $\mathcal A_t$.
\end{itemize}
\label{defn:H-null}
\end{defn}

The definition of $\mathcal A_t$ is the same as the one given in Section~\ref{sec:particle-types},
but that of $\mathcal N_t$ is slightly different ---
previously we have defined $\mathcal N_t$ by the condition that $O$ is annihilated by
the product of $h_a$'s that are supported on the annulus of thickness $t+w$. See Eq.~\eqref{eq:def-nullOP}.
The present definition clearly implies the previous one because the product of $h_{j(i)}$'s will annihilate $O$. Conversely,
\begin{lem}
Let $P$ be a finite product of some terms $h_j$ of the Hamiltonian supported on the annulus of thickness $t+w$.
For an operator $O \in \mathcal A_t$, if $OP = 0$, then $O \in \mathcal N_t$.
\label{lem:alg-decomposition}
\end{lem}
\begin{proof}
Let $P = h_{a_1} h_{a_2} \cdots h_{a_m}$.
One can rewrite $O = (1-P)O = \sum_{k=1}^m O_k$ where
\[
 O_k = h_{a_1}h_{a_2} \cdots h_{a_{k-1}}(1-h_{a_k}) O .
\]
It is clear that $O_k h_{a_k} = 0$.
\end{proof}

Now our
{\bf stable logical algebra condition} asserts that
when $10w \le t \le t' \le \rann / 10$
the inclusion $\mathcal A_t \hookrightarrow \mathcal A_{t'}$ induces an {\em isomorphism}
\[
\mathcal A_t/ \mathcal N_t \to \mathcal A_{t'} / \mathcal N_{t'}
\]
for the left and right annuli.
(Of course, the constant 10 is an arbitrary choice that is larger than 1.)
Note that the induced map is always well-defined because $\mathcal N_t$ is a subset of $\mathcal N_{t'}$.
That is, any equivalent operators of the quotient algebra $\mathcal A_t / \mathcal N_t$
are mapped to equivalent operators in $\mathcal A_{t'} / \mathcal N_{t'}$.

The induced map being surjective means that for any operator $O_{t'}$ of $\mathcal A_{t'}$
there is an operator $O_t$ equivalent to $O_{t'}$ modulo $\mathcal N_{t'}$
such that $O_t$ is supported on a thinner annulus of thickness $t$.
Recall that the operators of $\mathcal A_t$ are physically the string operators transporting anyons.
The equivalence relation by $\mathcal N_t$ is motivated by the fact that the string operators
should be deformed without affecting its action if the deformation sweeps a region of no anyons.
Thus, the existence of $O_t$ reflects the intuition that any string operator's action 
can always be achieved by acting on a narrow region.

The induced map being injective means that any nontrivial action ($\notin \mathcal N_t$)
of the string operator will still be nontrivial even if one has provided
with additional information that there is no anyon on an enlarged region ($\notin \mathcal N_{t'}$).
This is related to the fact that there are finitely many anyon labels,
and any anyon in the disk enclosed by the annulus
can be moved to a fixed finite region,
which can be chosen to be the center of the disk
without altering the far outside region.

\begin{remark}
\label{remark-nonstable-logical-algebra}
Indeed, the injectivity of the induced map $\mathcal A_t/ \mathcal N_t \to \mathcal A_{t'} / \mathcal N_{t'}$
may be false in a system with infinitely many particle types;
the stable logical algebra condition is not implied by the local topological order condition in general.
As a concrete example, we can consider an infinite stack of toric code layers.
Let us say that the layers are parallel to the $xy$-plane.
This state is three-dimensional,
but the particles can only move within each layer;
there are infinitely many particle types living on different layers.
An analog of the annulus in the three-dimensional space is a sphere with a wall of thickness $t$.
The algebra $\mathcal A_t$ and its ideal $\mathcal N_t$
are similarly defined.
Depending on the position of the layer,
the intersection of the layer and the wall can be either an annulus, a disk, or empty.
An $e$-particle projector $\pi_e$ acting on a layer $Y$
is an element of $\mathcal A_t / \mathcal N_t$,
which is zero if the intersection of $Y$ with the wall is a disk,
but nonzero if it is an annulus.
Suppose the layer $Y$ is close to the north pole,
but is not too close so that $\pi_e$ is nonzero.
As we increase the thickness of the wall,
the intersection of $Y$ with the wall changes its topology from an annulus to a disk.
Therefore, at some thickness $t' > t$, $\pi_e$ becomes zero in $\mathcal A_{t'} / \mathcal N_{t'}$.
The induced map $\mathcal A_t / \mathcal N_t \to \mathcal A_{t'} / \mathcal N_{t'}$ is not injective
in this example.
A similar situation is also found in the cubic code model~\cite{Haah2011Local}.

As far as we know, the only examples where this injectivity fails
exist in three or higher dimensions.
It is an interesting problem whether the stable logical algebra condition follows from the local topological order condition
in two dimensions.
\hfill $\lozenge$
\end{remark}

\subsection{Proof of Theorem~\ref{thm:invariant}}
\label{sec:proof}

We now prove the theorem.
Fix a ground state $\ket \psi$ of a Hamiltonian Eq.~\eqref{eq:local-H}
with local topological order.

\begin{defn}
Given a state $\ket \psi$,
an operator $O$ is said to be {\bf $s$-locally null} near an annulus
if it can be expressed as
\[
O = O_1 + \cdots + O_n,
\]
such that 
for each $O_i$ there exists a disk $D_i$ of radius at most $s$ meeting the annulus such that
\[
\Pi_{D_i} O_i = O_i \Pi_{D_i} = 0
\]
where
$\Pi_{D_i}$ is the projector onto the subspace support of the reduced density operator of $\ket \psi$ on the disk $D_i$.
In addition, we define two sets:
\begin{itemize}
 \item $\mathcal I_t$: the set of all operators $O$ supported on the annulus of thickness $t$ such that
      both $O$ and $O^\dagger$ are $(s,s)$-locally invisible for all $t/16 \le s \le t/4$.
 \item $\mathcal M_t$: the set of all $(t/2)$-locally null operators near the annulus of thickness $t$.
\end{itemize}
\end{defn}

Note that $\mathcal M_t$ \emph{is} a linear space, but $\mathcal I_t$ may \emph{not} be.
$\mathcal M_t$ is not guaranteed to be a subset of $\mathcal I_t$.
Nevertheless, we consider an equivalence relation on $\mathcal I_t$ defined by
\[
O \sim O' \text{ if and only if } O - O' \in \mathcal M_t
\]
for any $O,O' \in \mathcal I_t$.
The relation is reflexive ($O\sim O$),
symmetric ($O \sim O' \Leftrightarrow O' \sim O$),
and transitive ($O \sim O'$ and $ O' \sim O''$ imply $O \sim O''$).
We denote by 
\[
\mathcal I_t/ \mathcal M_t
\]
the set of all the equivalence classes.
$\mathcal I_t$ and $\mathcal M_t$ are defined purely in terms of the state,
while $\mathcal A_t$ and $\mathcal N_t$ depend on the Hamiltonian.

The locally null operators are defined similarly to the $H$-null operators. 
See Eq.~\eqref{eq:def-nullOP} and Lemma~\ref{lem:alg-decomposition}.
In particular, it holds that
any locally null operator makes the twist pairing to be zero.
Indeed, suppose $O$ is locally null, or more specifically
\[
 O \in \mathcal M_t,
\]
where $t$ is much smaller than the radius of the annulus.
Since the local projectors to annihilate $O$ will be supported on a thickened annulus,
we will need a left annulus $A_{2t}$ of thickness $2t$.
Recall that the left annulus and the right annulus intersect at two diamonds, one in the north and the other in the south.
Let $\Omega$ be the region of $A_{2t}$ that contains all points of the left annulus except those in the south diamond.
Similarly, let $\mho$ be the region of $A_{2t}$ that contains all points of the left annulus but those in the north diamond.
Let $\Pi_\Omega$ and $\Pi_\mho$ be, as in Lemma~\ref{lem:basic-projectors},
the projectors onto the subspace support of the reduced density operator of the ground state
on the region $\Omega$ and $\mho$, respectively.
It is easy to verify that
\[
 \Pi_\Omega \Pi_D = \Pi_\Omega, \quad \Pi_E \Pi_\mho = \Pi_\mho
\]
for any disk $D$ contained in $\Omega$ and any disk $E$ in $\mho$ by continued Schmidt decompositions.
(This is equivalent to say that $\ker \rho_{AB} \supseteq \ker \rho_A \otimes \rho_B$ for any bipartite density matrix $\rho_{AB}$.
This follows from $\Pi_A \Pi_B \rho_{AB} \Pi_B \Pi_A = \rho_{AB}$
where $\Pi_{A,B}$ are projectors such that $\ker \rho_{A,B} = \Pi_{A,B}^\perp$.)
Therefore,
\[
 \Pi_\Omega O \Pi_\mho = 0
\]
by the definition of $\mathcal M_t$.
Now,
\begin{align*}
\Pi_\Omega (O \infty O') &= (\Pi_\Omega O) \infty O'\\
(O \infty O') \Pi_\mho &= (O \Pi_\mho ) \infty O'
\end{align*}
for any $O'$ on the right annulus. See Figure~\ref{fig:annuli}.
Since $\Pi_\Omega \ket \psi  = \ket \psi = \Pi_\mho \ket \psi$,
we have 
\begin{equation}
\bra \psi O \infty O' \ket \psi = \bra \psi \Pi_\Omega (O \infty O') \Pi_\mho \ket \psi = 0.
\label{eq:locally-null-zero-twist-pairing}
\end{equation}
We conclude that any equivalent operators of $\mathcal I_t$ give the same value for the twist pairing.

\begin{lem}
An operator $O$ that commutes with every term of the Hamiltonian is
$(r,w)$-locally invisible with respect to the ground state $\ket \psi$ for any $r > 0$, where $w$ is the interaction range.
\label{lem:commuting-invisible}
\end{lem}
\begin{proof}
Pick any concentric disks $A,B$ of radius $r, r+w$, respectively.
Let $\ket \phi$ be a state vector with the same reduced density operator
on $B$ as $\ket \psi$.
By Lemma~\ref{lem:basic-projectors},
$h_j \ket \phi = \ket \phi$ whenever $h_j$ is supported inside $B$.
These include all $h_j$ whose support meet $A$.
By assumption,
we have $h_j O \ket \phi = O h_j \ket \phi = O \ket \phi$.
The local topological order condition implies that these equations determine
the reduced density matrix $\rho_A$ of $O\ket \phi$ on $A$,
and $\rho_A$ is the same as that of $\ket \psi$.
\end{proof}

\begin{lem}
Assume $t > 16w$.
$\mathcal A_t$ is a subset of $\mathcal I_{t'}$ for $t' \geq t$, 
and the inclusion map induces a well-defined map $\mathcal A_t / \mathcal N_t \to \mathcal I_{t'}/ \mathcal M_{t'}$.
\label{lem:ANtoIM}
\end{lem}
\begin{proof}
Lemma~\ref{lem:commuting-invisible} says that $\mathcal A_t \subseteq \mathcal I_t$.
(The constant 16 is because we have defined members of $\mathcal I_t$ 
to be $(t/16,t/16)$-locally invisible.)
We have to show that the induced map is well-defined,
i.e., that equivalent operators of $\mathcal A_t / \mathcal N_t$ 
are mapped to the same equivalence class of $\mathcal I_{t'}/ \mathcal M_{t'}$.
Suppose $O - O' \in \mathcal N_t$.
There exists a decomposition
\[
 O - O' = O_1 + \cdots + O_n,
\]
such that for each $O_i$ there is $h_{j(i)}$ on the annulus of thickness $t+w$ such that $O_i h_{j(i)} = h_{j(i)} O_i = 0$.
We can certainly choose a disk $D_i$ of radius $t'/2$ that contains the support $h_{j(i)}$ and meets the annulus.
Let $\Pi_{D_i}$ be as in Lemma~\ref{lem:basic-projectors}.
Since $h_{j(i)} \Pi_{D_i} = \Pi_{D_i}$, we see $O_i \Pi_{D_i} = \Pi_{D_i} O_i = 0$.
This means that $O \sim O'$.
\end{proof}

Conversely, we can map locally invisible operators into the logical algebra.
Consider the symmetrization (superoperator) $\phi_j$ for each $h_j$ defined by
\begin{equation}
 \phi_j : O \mapsto \frac12 ( O + (2h_j-1) O (2h_j -1) ),
\label{eq:sym-single}
\end{equation}
and the composition of all $\phi_j$'s
\begin{equation}
 \phi = \prod_j \phi_j.
\label{eq:sym-composition}
\end{equation}
Since $h_j$'s are commuting with one another, the order of the composition does not matter.
If $O$ commutes with $h_j$, then $\phi_j(O) = O$.
Therefore, applying $\phi$ to any operator enlarges its support only by $w$, the interaction range.
It will be used below that for any $h_j$,
\begin{equation}
 (2h_j-1)[ O - \phi_j(O) ](2h_j -1) = - [O- \phi_j(O)].
 \label{eq:sOphiOs}
\end{equation}

\begin{lem}
The symmetrization map $\phi$ of Eq.~\eqref{eq:sym-composition} 
defines a map from $\mathcal I_t$ to $\mathcal A_{t+w}$,
and for any $S, T \in \mathcal I_t$ we have
\[
 S(T - \phi(T)) \in \mathcal M_t, \quad (T-\phi(T))S \in \mathcal M_t.
\]
\label{lem:symmetrization}
\end{lem}
Setting $S=\mathrm{id}$, we see that any locally invisible operator can be 
``deformed'' to an operator that commutes with every term in the Hamiltonian
by enlarging its support slightly.
\begin{proof}
That $\phi : \mathcal I_t \to \mathcal A_{t+w}$ is already shown.
For clarity of notation,
assume that $h_j = h_1, \ldots, h_m$ are all the terms of the Hamiltonian
whose support overlap the annulus $A_t$ on which $O$ is supported.
So, $\phi(O) = (\phi_m \phi_{m-1} \cdots \phi_1)(O)$.

Consider, for each $h_j$ that meets the annulus $A_t$,
concentric disks $D_{j}^r$ of radius $r$,
with their center inside the support of $h_j$.
We will use the disks of radius $r = t/16, 2t/16, \ldots, 8t/16$.

Let $Q_j$ be the product of all terms $h_i$ of the Hamiltonian whose support is contained in $D_j^{3t/16}$.
The local topological order condition implies that for any state $\ket \phi$, the state $Q_j \ket \phi$
has the same reduced density matrix on $D_j^{2t/16}$ as the ground state.
Since $T$ is $(t/16,t/16)$-locally invisible, the reduced density matrix on $D_j^{t/16}$ of $T Q_j \ket \phi$
is the same as that of the ground state. 
In particular, $h_j T Q_j \ket \phi = T Q_j \ket \phi $
for any state $\ket \phi$, or equivalently,
\begin{equation}
 h_j T Q_j = T Q_j .
 \label{eq:hTQeqTQ}
\end{equation}

The disk $D_j = D_j^{t/2}$ certainly meets the annulus $A_t$.
By Lemma~\ref{lem:basic-projectors}, we have
\begin{equation}
 Q_j \Pi_{D_j} = \Pi_{D_j}.
 \label{eq:QPieqPi}
\end{equation}
Now, let $\phi_I$ be any composition of $\phi_i$'s,
and $s_j = 2h_j -1$.
\begin{align*}
 &S [ \phi_j \phi_I(T) - \phi_I(T) ] \Pi_{D_j} & \\
& = S [ \phi_j \phi_I(T) - \phi_I(T) ] Q_j s_j \Pi_{D_j} & \text{Eq.~\eqref{eq:QPieqPi}} \\
& =  S[ \phi_j \phi_I(T Q_j) - \phi_I(T Q_j) ] s_j \Pi_{D_j} &\text{$h_j$'s are commuting}\\
& =  S[ \phi_j \phi_I(s_j T Q_j) - \phi_I(s_j T Q_j) ] s_j \Pi_{D_j} & \text{Eq.~\eqref{eq:hTQeqTQ}}\\
& =  Ss_j [ \phi_j \phi_I(T) - \phi_I(T) ] Q_j s_j \Pi_{D_j} &\text{$h_j$'s are commuting}\\
& =  Ss_j [ \phi_j \phi_I(T) - \phi_I(T) ] s_j \Pi_{D_j} & \text{Eq.~\eqref{eq:QPieqPi}} \\
&= -  S[ \phi_j \phi_I(T) - \phi_I(T) ] \Pi_{D_j} & \text{Eq.~\eqref{eq:sOphiOs}}\\
& = 0
\end{align*}

Next, we show $[ \phi_j \phi_I(T) - \phi_I(T) ] S \Pi_{D_j} = 0$ by a similar calculation.
Let $Q_j'$ be the product of all terms $h_i$ of the Hamiltonian whose support is contained in $D_j^{t/2}$.
Since $S$ is $(3t/16, 3t/16)$-locally invisible and $Q_j'$ determines the reduced density matrix
on $D_j^{6t/16}$, we have
\begin{equation}
S \Pi_{D_j} = S Q_j' \Pi_{D_j} = Q_j S Q_j' \Pi_{D_j} = Q_j S \Pi_{D_j}.
\label{eq:SPieqQSPi}
\end{equation}
With Eq.~\eqref{eq:SPieqQSPi} in place of Eq.~\eqref{eq:QPieqPi},
one can similarly show that
\begin{align*}
& [ \phi_j \phi_I(T) - \phi_I(T) ] S \Pi_{D_j} = 0.
%& = [ \phi_j \phi_I(T) - \phi_I(T) ] Q_j s_j S \Pi_{D_j} \\
%& =  [ \phi_j \phi_I(T Q_j) - \phi_I(T Q_j) ] s_j S \Pi_{D_j} \\
%& =  [ \phi_j \phi_I(s_j T Q_j) - \phi_I(s_j T Q_j) ] s_j S\Pi_{D_j} \\
%& =  s_j [ \phi_j \phi_I(T) - \phi_I(T) ] Q_j s_j S \Pi_{D_j} \\
%& =  s_j [ \phi_j \phi_I(T) - \phi_I(T) ] s_j S \Pi_{D_j} \\
%&= -  [ \phi_j \phi_I(T) - \phi_I(T) ] S \Pi_{D_j}\\
%& = 0
\end{align*}
Since $S^\dagger$ and $T^\dagger$ are also locally-invisible,
the hermitian conjugate of the above computation is valid:
\begin{align*}
\Pi_{D_j} [ \phi_j \phi_I(T) - \phi_I(T) ] S = 0 , \quad \quad
\Pi_{D_j} S [ \phi_j \phi_I(T) - \phi_I(T) ] = 0.
\end{align*}
We have proved that
\begin{align*}
  [\phi_j \phi_I(O) - \phi_I(O)] S  \in \mathcal M_t, \quad \quad
  S [\phi_j \phi_I(O) - \phi_I(O)]  \in \mathcal M_t.
\end{align*}
Finally, using
\begin{equation}
\phi(T) - T = \sum_{j=1}^{m} \phi_{j}\phi_{j-1}\cdots \phi_1 (T) - \phi_{j-1}\phi_{j-2}\cdots \phi_1 (T)
\end{equation}
we see that
\[
 S (\phi(T) - T) \in \mathcal M_t, \quad \quad (\phi(T) - T) S \in \mathcal M_t.
\]
\end{proof}

\begin{lem}
Assume $3w < t$.
Suppose the canonical map $\mathcal A_{t+w}/ \mathcal N_{t+w} \to \mathcal A_{3t} / \mathcal N_{3t}$ 
induced by the inclusion $\mathcal A_{t+w} \hookrightarrow \mathcal A_{3t}$ is injective.
Then, the induced map $\bar \phi : \mathcal I_t / \mathcal M_t \to \mathcal A_{t+w} / \mathcal N_{t+w}$
by the symmetrization map $\phi$ of Eq.~\eqref{eq:sym-composition} is well-defined and injective.
\label{lem:IMtoAN}
\end{lem}
\begin{proof}
We first have to show that $\bar \phi$ is well-defined.
Suppose $O \sim O'$ for $O,O' \in \mathcal I_t$, i.e.,
\[
 O - O' = O_1 + \cdots + O_n
\]
where
for each $O_i$ there exists a disk $D_i$ of radius $\le t/2$ meeting the annulus of thickness $t$
such that $\Pi_{D_i} O_i = O_i \Pi_{D_i} = 0$.
Here, $\Pi_{D_i}$ is as in Lemma~\ref{lem:basic-projectors}.
Collect all $h_j$ that meet $D_i$ and let $P_{D_i}$ be the product of those.
By Lemma~\ref{lem:basic-projectors}, we have $\Pi_{D_i} P_{D_i} = P_{D_i}$.
Note that $P_{D_i}$ is supported on the annulus of thickness $3t$.

Since $\phi$ is linear,
\[
 \phi(O) - \phi(O') = \phi(O_1) + \cdots + \phi(O_n) .
\]
Since $h_j$'s are commuting,
\[
 \phi(O_i) P_{D_i} = \phi( O_i P_{D_i} ) = \phi( O_i \Pi_{D_i} P_{D_i} ) = 0.
\]
Thus, $\phi(O) - \phi(O') \in \mathcal N_{3t}$ by Lemma~\ref{lem:alg-decomposition}.
But, we already know $\phi(O) - \phi(O') \in \mathcal A_{t+w}$.
Since the kernel of the map
\[
\mathcal A_{t+w} \to \mathcal A_{3t} / \mathcal N_{3t}
\]
is $\mathcal A_{t+w} \cap \mathcal N_{3t} = \mathcal N_{t+w}$ by assumption,
it follows that $\phi(O) - \phi(O') \in \mathcal N_{t+w}$.
Therefore, $\bar \phi$ is well-defined.

To show $\bar \phi$ is injective, suppose $\phi(O) - \phi(O') \in \mathcal N_{t+w}$.
We have to show that $O- O' \in \mathcal M_t$.
By the argument in the proof of Lemma~\ref{lem:ANtoIM},
we see that $\mathcal N_{t+w} \subseteq \mathcal M_t$.
It remains to show $O - \phi(O) \in \mathcal M_t$ and $O' - \phi(O') \in \mathcal M_t$,
but this is already shown in Lemma~\ref{lem:symmetrization}.
The injectivity of the map $\bar \phi$ is proved.
\end{proof}

\begin{cor}
Suppose $\mathcal A_t / \mathcal N_t$ for various $20w < t < \rann / 20$
are all isomorphic by the inclusions $\mathcal A_{t_1} \hookrightarrow \mathcal A_{t_2}$ whenever $t_1 \le t_2$,
i.e., assume the stable logical algebra condition.
Then, the inclusion $\mathcal A_t \hookrightarrow \mathcal I_t$ induces a bijection
\[
 \mathcal A_t / \mathcal N_t \cong \mathcal I_t / \mathcal M_t.
\]
for $ 20 w < t < \rann / 60$.
\label{cor:bijection}
\end{cor}
\begin{proof}
We have a commutative diagram inferred from Lemma~\ref{lem:ANtoIM} and \ref{lem:IMtoAN}.
Our choice of the range $20w < t < \rann / 60$ is to ensure we can apply these lemmas.
\[
\xymatrix{
\mathcal A_t \ar[r]^{\iota} \ar[d] & \mathcal I_t \ar[r]^\phi \ar[d] & \mathcal A_{t+w}\ar[d] \\
\mathcal A_t/\mathcal N_t \ar[r]^{\bar \iota} & \mathcal I_t /\mathcal M_t \ar[r]^{\bar \phi} & \mathcal A_{t+w}/ \mathcal N_{t+w}
}
\]
The vertical arrows are the canonical maps from elements to equivalence classes that they represent.
The assumption implies $\bar \phi$ is well-defined and injective.
The upper line $\phi \circ \iota$ is just the inclusion map 
since the symmetrization acts as identity on the elements that already commute with $h_j$.
The assumption says $\bar \phi \circ \bar \iota$ is bijective.
It follows that $\bar \phi$ is surjective.
Therefore, $\bar \phi$ is bijective, and in turn, this means that $\bar \iota$ is bijective.
\end{proof}

\begin{cor}
Suppose Hamiltonians $H_1$ and $H_2$ of interaction range $w$
share a ground state, and satisfy our
local topological order and stable logical algebra conditions.
Then, the logical algebras $\mathcal A_t^{H_1} / \mathcal N _t^{H_1}$ and  $\mathcal A_t^{H_2} / \mathcal N _t^{H_2}$
are isomorphic as $C^*$-algebras for some $t$ (and hence for all $t$)
where $20w < t < \rann/60$,
by a map that preserves the twist pairing.
\label{cor:unique-logical-algebra}
\end{cor}
\begin{proof}
Every map in this subsection (the symmetrization Eq.~\eqref{eq:sym-composition} or the inclusion)
preserves hermitian conjugation.
Thus, it suffices to show that the composition $\Phi$
\begin{equation}
\Phi : 
\xymatrix{
 \mathcal A_t^{(1)} / \mathcal N _t^{(1)} 
\ar[r]^{ \bar \iota_1 } & 
 \mathcal I_{t} / \mathcal M_{t} 
\ar[r]^{ \bar \phi_2 } &
 \mathcal A_{t+w}^{(2)} / \mathcal N_{t+w}^{(2)}
\ar[r]^{ \bar \iota_2 } &
 \mathcal I_{t+w} / \mathcal M_{t+w} 
\ar[r]^{ \bar \phi_2' } &
 \mathcal A_{t+2w}^{(2)} / \mathcal N_{t+2w}^{(2)}
}
\end{equation}
preserves addition and multiplication, i.e., is an algebra-homomorphism.
Here,
the maps $\bar \iota_1, \bar \iota_2$ are by Lemma~\ref{lem:ANtoIM},
and $\bar \phi_2, \bar \phi_2'$ are by Lemma~\ref{lem:IMtoAN};
$w$ is the interaction range of $H_2$;
and the superscripts $(1),(2)$ refer to the Hamiltonians $H_1, H_2$.
The stable logical algebra condition and Corollary~\ref{cor:bijection}
imply that $\Phi$ is bijective.

The addition is readily preserved in $\Phi$ since every map is linear.
To show that the multiplication is preserved, suppose $x,y \in \mathcal A_t^{(1)}$.
Then, surely, $xy \in \mathcal A_t^{H_1}$.
We have to show that $\Phi([xy]) = \Phi([x])\Phi([y])$.
(The brackets are to distinguish representatives from their equivalence classes.)
Since every map is well-defined,
we may instead trace the images of $x,y,xy$
under inclusions and symmetrizations.
Under $\iota_2 \circ \phi_2 \circ \iota_1 : \mathcal A_t^{(1)} \to \mathcal I_{t+w}$,
the operators $x,y,xy$ are mapped to $\phi_2(x), \phi_2(y), \phi_2(xy) \in \mathcal I_{t+w}$, respectively.
Now, consider the equivalence relation supplied by $\mathcal M_{t+w}$.
Lemma~\ref{lem:symmetrization} applied to $\mathcal I_{t+w}$ implies
\begin{align*}
 \phi_2(xy) &\sim xy &&\because xy \in \mathcal A_t^{(1)} \subseteq \mathcal I_{t+w} , \\
 xy &\sim x \phi_2(y) &&\because x,y \in \mathcal A_t^{(1)} \subseteq \mathcal I_{t+w} ,\\
 x \phi_2(y) &\sim \phi_2(x)\phi_2(y) && 
 \because 
 \begin{cases} 
 x  \in \mathcal I_{t+w}, \\ 
 \phi_2(y) \in \mathcal A_{t+w}^{(2)} \subseteq \mathcal I_{t+w} 
 \end{cases} 
\end{align*}
where $\because$ stands for ``because.''
The transitivity of the relation $\sim$
implies
\[
\phi_2(xy) = \phi_2(x) \phi_2(y) + m
\]
for some element $m \in \mathcal M_{t+w}$.
We have then
\[
\phi_2' (\phi_2(xy)) = \phi_2'( \phi_2(x) \phi_2(y) ) + \phi_2'(m).
\]
Lemma~\ref{lem:IMtoAN} implies that $\phi_2'(m) \in \mathcal N_{t+2w}$.
Since an operator that already commutes with every term in the Hamiltonian
remains the same under further symmetrization,
we have $\phi_2' (\phi_2(xy)) = \phi_2(xy)$
and $\phi_2'( \phi_2(x) \phi_2(y) ) = \phi_2(x) \phi_2(y)$.
This proves $\Phi([xy]) = \Phi([x])\Phi([y])$,
and the $C^*$-isomorphism is established.

It remains to show that the isomorphism preserves the twist pairing.
Let $[O_i] \in \mathcal A^{H_i}_t / \mathcal N^{H_i}_t$ for $i=1,2$ be
the corresponding elements under the isomorphism.
We have shown that $[O_1]$ and $[O_2]$ are mapped by the inclusion 
to the same element in $\mathcal I_t / \mathcal M_t$.
This implies that $O_1 - O_2 \in \mathcal M_t$.
By the result of Eq.~\eqref{eq:locally-null-zero-twist-pairing},
this means that $O_1$ and $O_2$ give the same twist pairing.
\end{proof}
Note that we made use of two consecutive symmetrizations $\phi_2$ and $\phi_2'$,
which might seem redundant,
in order to treat the remainder term $m$.

\begin{proof}[Proof of Theorem~\ref{thm:invariant}]
We have proved the first statement in Corollary~\ref{cor:unique-logical-algebra}.
To prove the second statement,
we claim that the new Hamiltonian
\[
H' = W H W^\dagger = \sum_j W h_j W^\dagger
\]
for $W \ket \psi$ is commuting and locally topologically ordered,
and the corresponding algebra 
\[
\mathcal C'_{t'} = \mathcal A'_{t'} / \mathcal N'_{t'}
\]
is isomorphic to $\mathcal C_t$ when $t' = t+ R$ with an isomorphism being the conjugation by $W$.
Here, $\mathcal A'_{t'}$ and $\mathcal N'_{t'}$ are with respect to the new Hamiltonian $W H W^\dagger$.

Then, since $\mathcal C_t$ is stable with respect to $t$, so will $\mathcal C'_{t'}$.
This will mean that $W \ket \psi$ admits a locally topologically ordered Hamiltonian,
and the logical algebra is stable, so $\tilde S(W \ket \psi)$ will be defined.
Moreover, since the isomorphism $\mathcal C_t \to \mathcal C'_{t'}$ is induced by $W$,
the fundamental projectors $\pi_a$ of $\mathcal C_{t}$ will represent fundamental projectors
of $\mathcal C'_{t'}$ as $W \pi_a W^\dagger$. Therefore,
\[
 \bra \psi W^\dagger (W \pi_a W^\dagger) \infty (W \pi_b W^\dagger) W \ket \psi = \bra \psi \pi_a^L \infty \pi_b^R \ket \psi 
\]
by Lemma~\ref{lem:inv-twist-prod}, and we will complete the proof of the theorem.

The terms of the new Hamiltonian $WHW^\dagger$ is clearly commuting with one another,
since the conjugation is an algebra-automorphism.
To show that it is locally topologically ordered,
we use the observation made in the proof of Lemma~\ref{lem:evolved-invisibility}.
Given a disk $D$, let $D'$ be $(w+R)$-ball of $D$.
Let $P_{D'}$ be the product of all terms $h_j$ of $H$ that are supported on $D'$,
and define $P'_{D'} = W P_{D'} W^\dagger$.
$P'_{D'}$ is supported on $(w+2R)$-ball of $D$.
Let $\ket \phi$ be any state such that $P'_{D'} \ket \phi = \ket \phi$.
We claim that the reduced density matrix on $D$ of $\ket \phi$ is the same as that of $W\ket \psi$ on $D$.
Since $P_D W^\dagger \ket \phi = W^\dagger \ket \phi$,
using the local topological order condition of $H$, 
we know $W^\dagger \ket \phi$ has the same reduced density matrix
on the $R$-ball of $D$ as $\ket \psi$.
Using the fact that $W$ is a quantum circuit of range $R$,
the reduced density matrix on $D$ of $\ket \phi = W(W^\dagger \ket \phi)$
is the same as that of $W \ket \psi$.
This proves the claim.

To show that $\mathcal C'_{t'}$ is isomorphic to $\mathcal C_t$, consider the conjugation by $W$
\[
 \omega: \mathcal A_t \xrightarrow{W \cdot W^\dagger} \mathcal A'_{t+R} .
\]
It is straightforward to check that $W \mathcal N_t W^\dagger \subseteq \mathcal N'_{t+R}$,
so
\[
 \bar \omega : \mathcal C_t \xrightarrow{W \cdot W^\dagger} \mathcal C'_{t+R}
\]
is induced and well-defined. Similarly, consider a map
\[
 \bar \omega^\dagger : \mathcal C'_{t+R} \xrightarrow{W^\dagger \cdot W} \mathcal C_{t+2R}.
\]
Certainly, $\omega^\dagger \circ \omega : \mathcal A_{t} \to \mathcal A_{t+2R}$ is
the inclusion, so $\bar \omega^\dagger \circ \bar \omega$ is an isomorphism by our stable logical algebra assumption.
In particular, $\bar \omega^\dagger$ is surjective.
We claim $\bar \omega^\dagger$ is injective as well.
Suppose $\omega^\dagger( O' ) = W^\dagger O' W \in \mathcal N_{t+2R}$.
By the definition of $\mathcal N_{t+2R}$, 
\[
 W^\dagger O' W = O_1 + \cdots + O_n
\]
and there exist $h_{j(i)}$ supported on the annulus of thickness $t+2R+w$
such that $O_i h_{j(i)} = h_{j(i)} O_i = 0$. The decomposition can be written equivalently as
\[
 O' = W O_1 W^\dagger + \cdots + W O_n W^\dagger .
\]
Each $W O_i W^\dagger$ is annihilated by $W h_{j(i)} W^\dagger$ which is supported
on the annulus of thickness $t+3R+w = t+R +w'$, where $w' = w+2R$ is the interaction range of $H'$.
This means that $O' \in \mathcal N'_{t+R}$,
and the injectivity of $\bar \omega^\dagger$ is proved.
Therefore, $\bar \omega^\dagger$ is bijective.
We conclude that $\bar \omega$ is also bijective, and $\mathcal C_t$ and $\mathcal C'_{t+R}$ are isomorphic.
\end{proof}

\section{Discussion}

We have defined an invariant matrix $\tilde S$ under local unitary transformations
for ground states of a particular class of commuting Hamiltonians.
$\tilde S$ is naturally motivated by braiding of quasi-particles
and its modular tensor category description,
but we have avoided dealing with the axioms of the categorical description,
by analyzing operator algebras that stabilize the state from the beginning.
Our definition of $\tilde S$ matrix is in terms of the Hamiltonian,
complemented by a proof that $\tilde S$ is independent of the Hamiltonian
as long as it belongs to a class specified by the local topological order
and stable logical algebra conditions.

It is certainly desirable to have an invariant that is 
manifestly defined in terms of the bulk of the state.
The topological entanglement entropy could be a good exemplary quantity,
but we do not yet have a rigorous answer when this quantity is invariant.
Our notion of locally invisible operators would give much more information
about the state, but they generally do not form an algebra,
which makes it hard to extract useful information out.
Our local topological order and stable logical algebra assumptions
constitute a sufficient condition for the locally invisible operators
to form an algebra.
Perhaps, in two dimensions the stable logical algebra assumption
may follow from the local topological order assumption,
as the latter seems to imply number of particle types in 2D is bounded.
This is an interesting open problem.

Another interesting problem is to determine 
when the unitarity of the $S$ matrix (the modularity) can be derived.
Note that this cannot be proved unless some homogeneity is assumed.
For example, if an annulus happens to include a gapped boundary,
then the logical algebra may have different $\mathbb C$-dimension
from that on an annulus in the bulk.
Then, the $\tilde S$ matrix between the two logical algebras is not necessarily square.

We have mainly discussed two-dimensional lattices, 
but a higher dimensional generalization is straightforward;
one can replace the annulus with spheres with a thick enough wall,
or more generally consider any embedding of lower dimensional manifolds.
In three dimensions for example, one can consider
a sphere and a circle that are intersecting at a pair of distant points.
This corresponds to braiding of point-like particles around string-like excitations.

\begin{acknowledgments}
The author is grateful to Sergey Bravyi and Aram Harrow for useful discussions,
and acknowledge the hospitality of
Simons Institute for Theory of Computing, Berkeley, California,
where a part of this work was done.
The author also thanks Zhenghan Wang for informing Ref.~\cite{Naidu2007}.
This work is supported by MIT Pappalardo Fellowship in Physics.
\end{acknowledgments}

\end{document}